\documentclass[11pt,a4paper]{article}
\usepackage[margin=1in]{geometry}
\usepackage{tikz}
\usetikzlibrary{positioning,arrows.meta}
\usepackage{subcaption}

\usepackage[english]{babel}
\usepackage{algpseudocode}
\usepackage{makecell}
\usepackage[english]{babel}
\usepackage{array} 
\usepackage{booktabs} 
\usepackage{authblk} 

\usepackage[ruled,linesnumbered,vlined]{algorithm2e}
\usepackage{algpseudocode}
\usepackage{amssymb}
\usepackage{mathtools}
\DeclarePairedDelimiter\floor{\lfloor}{\rfloor}
 
\usepackage{amsmath}
\usepackage{amsthm}
\usepackage{graphicx}
\usepackage[colorlinks=true, allcolors=blue]{hyperref}
\usepackage{natbib}
\usepackage{framed}
\definecolor{shadecolor}{gray}{0.9}
\usepackage{pgfplots}
\pgfplotsset{compat=1.17}
\usepackage{tikz}
 \usetikzlibrary{patterns}
\usepackage{pgfplotstable}
    \pgfplotsset{
    name nodes near coords/.style={
        every node near coord/.append style={
            name=#1-\coordindex,
            alias=#1-last,
        },
    },
    name nodes near coords/.default=coordnode
    }

\usepackage{multirow}

\usepackage[table]{xcolor} 

\newtheorem{theorem}{Theorem}[section]
\newtheorem{definition}[theorem]{Definition}

\newtheorem{corollary}[theorem]{Corollary}
\newtheorem{lemma}[theorem]{Lemma}

\newtheorem{example}[theorem]{Example}
\newtheorem{observation}[theorem]{Observation}

\newcommand{\bX}{\mathbf{X}}

\newcommand{\MMS}{\mathsf{MMS}}

\definecolor{ruleblue}{RGB}{70,120,200}
\definecolor{shadegray}{gray}{0.9}
  
\newcommand{\bluerule}{\arrayrulecolor{ruleblue}\specialrule{0.8pt}{0pt}{0pt}\arrayrulecolor{black}}

\makeatletter
\renewcommand{\and}{\end{tabular}\hspace{3em}\begin{tabular}[t]{c}}
\makeatother

\title{When Maximum Nash Welfare Becomes Strongly Fair: \\
Bi-valued Goods\thanks{The authors are ordered alphabetically.}}

\author[1]{Zehan Lin}
\author[1]{Xiaowei Wu}
\author[2]{Shengwei Zhou}
\affil[1]{University of Macau, \texttt{\{yc47490,xiaoweiwu\}@um.edu.mo}}
\affil[2]{Nanyang Technological University, \texttt{s.arthur.zhou@gmail.com}}

\date{December 2025}

\begin{document}

\maketitle
\begin{abstract}
    For the fair allocation of indivisible goods, the milestone work of Caragiannis et al. (2019) revolutionized the understanding of Maximum Nash Welfare (MNW) allocations by revealing their ``unreasonable'' ability to balance fairness (EF1) and efficiency (PO) for general additive functions. 
    Subsequent work established that MNW allocations simultaneously satisfy EFX and GMMS (and consequently MMS and PMMS) for binary instances (Amanatidis et al. 2021 and Barman et al. 2018), and satisfy EFX for bi-valued instances (Amanatidis et al. 2021).
    In this paper, we revisit the (personalized) bi-valued instances and provide a tight characterization of the approximation guarantees of MNW allocations with respect to EFX, PMMS, and GMMS, showing that these guarantees are significantly stronger than what is known for general valuations. 
    Furthermore, we demonstrate that even for the slightly more general setting of tri-valued instances, the fairness guarantees of MNW deteriorate significantly, which establishes a sharp boundary on the instances for which MNW allocations are strongly fair.
\end{abstract}

\section{Introduction}
The problem of fair division involves allocating a set of indivisible goods $M$ among a group of agents $N$. 
A primary objective in the problem is ensuring fairness among agents.
One of the most well-studied fairness notions is envy-freeness (EF)~\cite{foley1966resource}, where no agent prefers another agent's bundle over her own.
Unfortunately, EF allocations are not guaranteed to exist when items are indivisible.
Therefore, the literature has explored alternative fairness paradigms.

A well-studied notion, \emph{envy-freeness up to one item} (EF1)~\cite{conf/sigecom/LiptonMMS04}, requires that any envy an agent has toward another can be eliminated by removing some item from the envied agent's bundle. 
A substantially stronger requirement is \emph{envy-freeness up to any item} (EFX)~\cite{journals/teco/CaragiannisKMPS19}, which demands that envy disappears after removing any item. 
Beyond these comparison-based notions, several share-based notions set explicit utility thresholds for agents.
One of the most prominent is the \emph{Maximin Share} (MMS)~\cite{conf/bqgt/Budish10}, which represents the value an agent can guarantee by partitioning the items into $n$ bundles and receiving the least-preferred one.
The notion of \emph{Pairwise Maximin Share} (PMMS)~\cite{journals/teco/CaragiannisKMPS19} applies this idea to pairs of agents, and requires that for any pair of agents, an agent's value is at least what she could guarantee by re-partitioning the union of their bundles. 
Due to this pairwise nature, PMMS is conceptually closer to envy-freeness, notably serving as a sufficient condition for EFX$_+$\footnote{EFX$_+$ is a slight relaxation of EFX where envy can be eliminated by removing any good with positive value from the envied agent's bundle. For general additive valuations, the two notions are equivalent under the assumption of non-degeneracy~\cite{journals/tcs/AmanatidisBFHV21}.}.
As a strict generalization of both MMS and PMMS, the \emph{Groupwise Maximin Share} (GMMS)~\cite{conf/aaai/BarmanBMN18} extends this principle and requires an MMS guarantee for every subgroup of agents.

However, fairness alone is insufficient.
A desirable allocation rule also needs to ensure economic efficiency. 
One standard benchmark for efficiency is \emph{Pareto optimality} (PO), which requires that no other allocation can strictly increase the utility of one agent without decreasing that of another.
Unfortunately, fairness and efficiency are often at odds, as enforcing strict fairness guarantees can fundamentally conflict with the requirements of Pareto efficiency.
A central challenge in fair division is thus to balance these conflicting goals, identifying allocation rules that satisfy strong fairness guarantees while simultaneously achieving PO.
\medskip

\paragraph{Maximum Nash Welfare Allocations.}
The Maximum Nash Welfare (MNW) rule (maximizing the product of valuations $\prod_i v_i(X_i)$) serves as a cornerstone in the theory of fair division, celebrated for its unique ability to balance efficiency and fairness.
In the divisible setting, the MNW allocation yields a Competitive Equilibrium from Equal Incomes (CEEI) solution, which satisfies both EF and PO~\cite{varian1973equity}. 
Its impact is perhaps most profound in the indivisible setting: Caragiannis et al.~\cite{journals/teco/CaragiannisKMPS19} termed the rule ``unreasonably fair'' for its ability to simultaneously guarantee EF1 and PO for general additive valuations\footnote{A valuation function is additive if the value of a bundle is the sum of the value of items in the bundle. Unless otherwise specified, we focus on the setting of additive valuation functions in this paper.}.

\medskip

However, the reach of MNW's ``unreasonable fairness'' is limited when evaluated against other benchmarks of fairness in the general additive setting. 
Here, its guarantees are suboptimal: it fails to yield any constant approximation for EFX~\cite{journals/tcs/AmanatidisBFHV21}, provides an approximation of MMS that vanishes as the number of agents increases, and achieves only a $0.618$-PMMS~\cite{journals/teco/CaragiannisKMPS19}.
Conversely, the rule exhibits surprising power in the restricted domain of \emph{binary} valuations (values in $\{0,1\}$), where it implies exact EFX, GMMS (and consequently, MMS and PMMS) simultaneously\footnote{Note that MNW requires careful tie-breaking to satisfy EFX, while EFX implies GMMS for binary instances.}~\cite{journals/tcs/AmanatidisBFHV21, conf/aaai/BarmanBMN18}.
This creates a sharp contrast: MNW serves as a unified solution for binary instances but performs poorly for general cases.
The vast landscape between these two extremes remains largely unexplored, with one notable exception being the EFX guarantee for bi-valued instances~\cite{journals/tcs/AmanatidisBFHV21}.
Therefore, in this paper, we investigate how far the strong fairness guarantees of MNW extend beyond the binary setting before degenerating into the poor approximations observed in the general additive domain.

\subsection{Our Contribution}

In this work, we focus on (different settings of) bi-valued instances, and show that the MNW allocations provide much stronger approximation guarantees with respect to the fairness notions of EFX, PMMS, and GMMS.
We summarize the approximation ratios in Table~\ref{tab:mnw-fairness}.

\renewcommand{\arraystretch}{1.3} 
\begin{table}[htbp]
\centering
\setlength{\tabcolsep}{10pt}
\caption{Fairness guarantees of MNW allocations under different settings. The results with asterisk ($*$) were established by Amanatidis et al.~\cite{journals/tcs/AmanatidisBFHV21}, while results with octothorpe ($\#$) were by Caragiannis et al.~\cite{journals/teco/CaragiannisKMPS19}.
Bi-valued structures enable substantially stronger guarantees.}
\label{tab:mnw-fairness}

\begin{tabular}{l c c c}
\arrayrulecolor{ruleblue}\specialrule{1.2pt}{0pt}{0pt}\arrayrulecolor{black}

\rowcolor{shadegray}
 & \textbf{EFX} & \textbf{PMMS} & \textbf{GMMS} \\
\bluerule

\textbf{General/Tri-valued} {\scriptsize (Thm.~\ref{theorem: trivalue})} & $0~^*$  & $0.618~^\#$  & ${2}/{(1+\sqrt{4n-3})}~^\#$  \\

\rowcolor{shadegray}
\textbf{Personalized Bi-valued\quad} & 
$2/3$ {\scriptsize (Thm.~\ref{theorem:pb_mnwefx})} & 
$2/3$ {\scriptsize (Thm.~\ref{theorem:pb_mnwpmms})} & 
$2/3$ {\scriptsize (Thm.~\ref{theorem:pb_mnwmms})} \\

\textbf{Bi-valued} & \multirow{2}{*}{$1~^*$} & $6/7$ {\scriptsize (Thm.~\ref{theorem:bi-mnwpmms})} & \multirow{2}{*}{$2/3$ {\scriptsize (Thm.~\ref{theorem:pb_efx_imply_2/3_MMS})}} \\

\textbf{Factored Bi-valued} & & $1$ {\scriptsize (Thm.~\ref{theorem:fb-mnwpmms})} & \\

\arrayrulecolor{ruleblue}\specialrule{1.2pt}{0pt}{0pt}\arrayrulecolor{black}
\end{tabular}
\end{table}

Our contributions are fourfold:

\paragraph{Tight Guarantees in the Bi-valued Hierarchy.}
We analyze MNW across a hierarchy of three settings.
For the most general \emph{personalized bi-valued} setting, we have $v_i(g) \in \{1,k_i\}$ for each agent $i$ and item $g$, where $k_i > 1$ could be different across different agents.
The \emph{bi-valued} instances assume that $k_i = k$ for all $i$, while the \emph{factored bi-valued} further requires that $k$ is an integer.
In sharp contrast to the limits of general instances, we show that within the bi-valued landscape, MNW achieves substantially stronger, constant-factor guarantees.
As summarized in Table~\ref{tab:mnw-fairness}, the approximation ratios of MNW to the fairness notions of EFX, PMMS, and GMMS are at least $2/3$.
Moreover, all bounds we establish are tight, i.e., for each setting and fairness notion, there exists an instance that admits an MNW allocation whose approximation ratio is not better than what we have established.

\paragraph{Characterizing the Boundary: Tri-valued Instances.}
We further probe the limits of MNW in \emph{tri-valued instances}, where $v_i(g) \in \{a, b, c\}$ for any agent $i$ and item $g$.
By constructing explicit counterexamples, we show that the fairness guarantees of MNW for EFX, PMMS, and GMMS deteriorate significantly, matching the negative results known for general additive valuations.
We prove that in this setting, MNW fails to yield any constant-factor approximation for EFX and offers a GMMS guarantee that vanishes as the number of agents increases (see Section~\ref{section:tri_upper} for details).
These findings fill the gap for the fairness guarantee of MNW between binary and general instances, and effectively mark the boundary of possibility: tri-valued instances represent the threshold where the fairness guarantees of MNW collapse.

\paragraph{Implications to Approximate Fair Allocations.}
As an implication of the fairness guarantees, our results improve upon the state-of-the-art bounds for approximating PMMS and GMMS for bi-valued instances, e.g., from $0.781$~\cite{kurokawa2016fair} to $6/7$ and from $0.618$~\cite{amanatidis2026improved} to $2/3$, respectively\footnote{We remark that the results of~\cite{kurokawa2016fair} and~\cite{amanatidis2026improved} also hold for general additive instances.}.
Another contribution of our work lies in the realm of simultaneous fairness guarantees, an area of growing interest~\cite{journals/tcs/AmanatidisMN20,conf/aaai/FeldmanMP24,conf/aaai/AkramiR25a,conf/sigecom/AshuriG25}.
Our results explicitly demonstrate the existence of allocations that simultaneously satisfy the approximate fairness as listed in Table~\ref{tab:mnw-fairness}, alongside PO.
For a comparison, for general additive functions, Amanatidis et al.~\cite{journals/tcs/AmanatidisMN20} proposed algorithms for computing allocations that are $0.618$-EFX, $2/3$-PMMS, and $0.553$-GMMS, and Ashuri et al.~\cite{conf/sigecom/AshuriG25} showed the existence of allocations that are $1/2$-EFX, $2/3$-PMMS, and $4/7$-GMMS.
However, these allocations are not guaranteed to be PO.
Our work demonstrates that the bi-valued structure enables a significant improvement across these dimensions, alongside PO.

\paragraph{On the Unique Position of MNW.}
Finally, our results further substantiate the unique position of MNW as the ``unreasonably fair'' rule, particularly when compared to other natural egalitarian objectives.
We first compare MNW with the leximin rule, another prominent egalitarian benchmark\footnote{The leximin rule selects an allocation that maximizes the utility of the worst-off agent, then explicitly maximizes the utility of the second worst-off agent, and so on, in a lexicographical manner.}.
While leximin and MNW are equivalent for binary instances~\cite{conf/ijcai/0001R20}, their performance diverges sharply in the bi-valued domain.
Leximin yields an approximation ratio of $1/k$ for EFX and GMMS, and $2/(k+1)$ for PMMS (see Appendix~\ref{section:leximin}), both of which vanish as $k$ grows.
Furthermore, we compare MNW with the broader family of maximum $p$-mean welfare allocations\footnote{For an allocation $\mathbf{X}$, its $p$-mean welfare is $\left(\frac{1}{n}\sum_{i\in N}v_i(X_i)^p\right)^{1/p}$ for $p\neq 0$, and is $\left(\prod_{i\in N}v_i(X_i)\right)^{1/n}$ for $p=0$.}. 
Within this parameterized family, MNW mathematically coincides with the limit case of $p = 0$. 
This specific point is structurally exceptional as for any fixed $p \neq 0$, maximum $p$-mean welfare allocations fail to guarantee any constant approximation for EFX, PMMS, or GMMS for bi-valued instances (see Appendix~\ref{section:pmean}).
These comparisons underscore the unique standing of MNW among these natural egalitarian objectives.

\subsection{Related Work}

Due to the vast literature on fair division, we focus only on the results most closely related to our work.
For a more comprehensive review, please refer to the survey by Amanatidis et al.~\cite{journals/ai/AmanatidisABFLMVW23}.


\paragraph{EFX Allocations.} 
The existence of EFX allocations is a central problem within the fair division community. 
Recent work has established negative results for broad valuation classes, demonstrating that EFX allocations are not guaranteed to exist under submodular valuations~\cite{journals/corr/abs-2604-18216, journals/corr/abs-2605-06451}. 
However, whether EFX allocations always exist for additive valuations for goods remains a major open problem.\footnote{Recent work~\cite{journals/corr/abs-2606-08872} showed that EFX allocations for chores may not exist even for four additive agents.} 
On the positive side, a series of works have made progress in establishing existence guarantees under various restricted settings.
These include instances with a small number of agents~\cite{journals/siamdm/PlautR20, journals/jacm/ChaudhuryGM24, journals/ior/AkramiACGMM25} and specific valuation classes, such as identical~\cite{journals/siamdm/PlautR20}, lexicographic~\cite{conf/aaai/HosseiniSVX21}, bi-valued~\cite{journals/tcs/AmanatidisBFHV21}, personalized bi-valued~\cite{journals/corr/abs-2507-18251, journals/corr/abs-2507-14957} or those with limited types~\cite{journals/dam/Mahara23, journals/corr/abs-2410-13580}. 
Progress has also been made in structured domains, including graph-based allocations~\cite{conf/sigecom/0001FKS23} and instances with only two item types~\cite{conf/soda/GorantlaMV23}. 
Regarding approximation, the initial $1/2$-EFX guarantee~\cite{conf/ijcai/Chan00W19} was improved to $0.618$ by Amanatidis et al.~\cite{journals/tcs/AmanatidisMN20}. 
This remains the state-of-the-art for general additive valuations, though improvements to $2/3$ are possible in special cases~\cite{conf/sigecom/AmanatidisFS24}.

\paragraph{MMS, PMMS, and GMMS Allocations.}
Since exact MMS allocations are not guaranteed to exist~\cite{journals/jacm/KurokawaPW18}, research has focused on approximations. The initial non-trivial $2/3$-MMS guarantee~\cite{journals/jacm/KurokawaPW18,journals/teco/BarmanK20,journals/talg/AmanatidisMNS17} was improved to $3/4+o(1)$ by a series of works~\cite{conf/sigecom/GhodsiHSSY18,journals/ai/GargT21,conf/ijcai/AkramiGST23,conf/soda/AkramiG24}, with recent breakthroughs pushing the ratio to $10/13$~\cite{journals/corr/abs-2510-10423} and $7/9$~\cite{huang2025fptas}.
Complementing these general approximation results, researchers have also investigated exact existence within restricted domains, e.g., for bi-valued valuations~\cite{conf/atal/EbadianP022,Feige2022MaximinFA}.
Regarding PMMS, recent work has established that PMMS allocations need not exist for general additive valuations~\cite{journals/corr/abs-2609-06282,journals/corr/abs-2609-08954}, although existence is guaranteed for identical valuations~\cite{journals/tcs/DaiGMGXZ24}.
PMMS is known to imply EFX$_+$.
Conversely, Amanatidis et al.~\cite{conf/ijcai/AmanatidisBM18} proved that EFX implies $2/3$-PMMS, while EF1 implies $1/2$-PMMS.
The current best approximation guarantee for PMMS is $0.781$~\cite{kurokawa2016fair}.
By definition, GMMS is strictly stronger than both MMS and PMMS.
Consequently, the non-existence of exact MMS allocations implies that exact GMMS allocations are not guaranteed to exist either.
Barman et al.~\cite{conf/aaai/BarmanBMN18} established the existence of $1/2$-GMMS allocations. 
This ratio was later improved to $4/7$ independently by Amanatidis et al.~\cite{journals/tcs/AmanatidisMN20} and Chaudhury et al.~\cite{journals/siamcomp/ChaudhuryKMS21}, who gave algorithms for computing $4/7$-GMMS allocations.
Most recently, Amanatidis et al.~\cite{amanatidis2026improved} further advanced the ratio to $(\phi-1)$-GMMS (where $\phi\approx 1.618$ is the golden ratio), and $(\sqrt{10}-1)/3$ for three agents.
Qin~\cite{journals/corr/abs-2608-30267} established tight $10/17$-MMS guarantees for both EFX and PMMS allocations under general additive valuations.
Applying these guarantees to every subset directly implies that both notions also satisfy $10/17$-GMMS.

\paragraph{Computation and Approximation of MNW.}
A significant body of literature focuses on the computation and approximation of MNW allocations.
As shown by Lee~\cite{journals/ipl/Lee17}, maximizing Nash welfare is APX-hard for additive valuations. 
This hardness has motivated a line of work on approximation algorithms, improving from a $2$-approximation~\cite{journals/jair/McGlaughlinG20} to the state-of-the-art factor of $1.45$~\cite{conf/sigecom/BarmanKV18, conf/icalp/Feng024, journals/corr/abs-2607-13340}.
The state-of-the-art hardness result is by Viswanathan~\cite{corr/abs-2605-27098}, who showed a lower bound of $1.076$.
Constant approximations of Nash welfare have also been established for the weighted settings~\cite{conf/icalp/Feng024,conf/stoc/FengH0025}, and valuations beyond additive~\cite{conf/stoc/GargHLVV23,conf/stoc/DobzinskiLRV24, journals/corr/abs-2504-09669}.
Finally, the concept of MNW has been applied to mixed settings by Nishimura and Sumita~\cite{journals/mss/NishimuraS25} and weighted agents by Suksompong and Teh~\cite{journals/mss/SuksompongT22}, establishing fairness guarantees in these generalized contexts.

\section{Preliminaries}

We study the problem of fairly allocating a set of $m$ indivisible goods, denoted by $M$, to a group of $n$ agents $N$.
We refer to a subset of items $B \subseteq M$ as a bundle.
Each agent $i \in N$ is associated with an additive valuation function $v_i: 2^M \to \mathbb{R}^+ \cup \{0\}$, which assigns a value to every bundle of items.
We use $\textbf{v} = (v_1, \dots, v_n)$ to denote the valuation functions of all agents.
For any bundle $B \subseteq M$ and item $g \in M$, we use $B + g$ to denote $B \cup \{g\}$, and $B - g$ to denote $B \setminus \{g\}$.
An allocation $\bX = (X_1, \dots, X_n)$ is an ordered partition of the item set $M$ into $n$ disjoint bundles, where agent $i$ receives bundle $X_i$. 
For any integer $t \geq 1$, we use $[t]$ to denote the set $\{1, \dots, t\}$. 
Given an instance $\mathcal{I} = (N, M, \textbf{v})$, the objective is to compute an allocation $\bX$ that is fair and efficient to all agents.

To evaluate the fairness of allocations, we introduce several fairness notions below.


\begin{definition}[EF1]
An allocation \textbf{X} is envy-free up to one item (EF1) if for any agents $i, j \in N$, either $X_j = \emptyset$, or there exists an item $g \in X_j$ such that
\begin{equation*}
    v_i(X_i) \geq v_i(X_j - g).
\end{equation*}
\end{definition}

\begin{definition}[$\alpha$-EFX]
For any $0 \leq \alpha \leq 1$, an allocation $\bX$ is $\alpha$-approximate envy-free up to any item ($\alpha$-EFX) if for any agents $i, j \in N$, either $X_j = \emptyset$, or for any $g \in X_j$, we have
\begin{equation*}
    v_i(X_i) \geq \alpha \cdot v_i(X_j - g).
\end{equation*}
When $\alpha = 1$, the allocation is EFX.   
\end{definition}


\begin{definition}[Maximin Share]
    Given a set of items $M$ and $n$ agents, the maximin share (MMS) value of agent $i$ is defined as
    $$
    \MMS_i (M,n) :=\max_{\bX \in \Pi_n(M)} \min_{j \in N}  \ \left\{v_i(X_j)\right\},
    $$
    where $\Pi_n(M)$ denotes the set of all $n$-partitions of $M$.
\end{definition}

We use $\MMS_i$ to denote $\MMS_i(M,n)$ when $M$ and $n$ are clear from the context.

\begin{definition}[$\alpha$-MMS]
    For any $0 \leq \alpha \leq 1$, an allocation $\bX$ is $\alpha$-approximate maximin share ($\alpha$-MMS) if for every agent $i \in N$, it holds that
    $$
    v_i(X_i) \geq \alpha \cdot \MMS_i.
    $$
    Specifically, when $\alpha = 1$, the allocation is an \emph{MMS} allocation.
\end{definition}

\begin{definition}[$\alpha$-PMMS]
    For any $0 \leq \alpha \leq 1$, an allocation $\bX$ is $\alpha$-approximate pairwise maximin share ($\alpha$-PMMS) if for every pair of agents $i$ and $j$, it holds that 
    $$
    v_i(X_i) \geq \alpha \cdot \MMS_i(X_i \cup X_j,2).
    $$
    Specifically, when $\alpha=1$, the allocation is a \emph{PMMS} allocation.
\end{definition}

It has been shown that any PMMS allocation also satisfies EFX$_+$\footnote{For bi-valued instances, EFX is equivalent to EFX$_+$.}~\cite{journals/teco/CaragiannisKMPS19}.
However, the implication breaks down in the approximation regime~\cite{conf/ijcai/AmanatidisBM18}, which motivates us to explore a complete characterization for each fairness notion independently.

\begin{definition}[$\alpha$-GMMS] An allocation $\bX$ is $\alpha$-approximate groupwise maximin share ($\alpha$-GMMS) if for every subset of agents $N' \subseteq N$ and any agent $i \in N'$, 
$$
v_i(X_i) \geq \alpha \cdot \MMS_i(\bigcup_{j \in N'} X_j, |N'|).
$$
Specifically, when $\alpha=1$, the allocation is a \emph{GMMS} allocation.
\end{definition}
By definition, GMMS is stronger than both MMS and PMMS. 
Furthermore, $\alpha$-GMMS can directly imply $\alpha$-MMS and $\alpha$-PMMS.
Note that all the fairness notions mentioned above are scale-free, i.e., if an allocation satisfies one of these notions, then it remains to satisfy the same notion if we rescale any valuation function.

Besides fairness, a desirable allocation should also guarantee efficiency, for which we consider the standard benchmark of Pareto optimality. 
\begin{definition}[Pareto Optimal]
    An allocation $\bX'$ \emph{Pareto dominates} another allocation $\bX$ if $v_i(X_i') \geq v_i(X_i)$ for all $i \in N$ and the inequality is strict for at least one agent. 
    An allocation $\bX$ is \emph{Pareto optimal} (PO) if $\bX$ is not Pareto dominated by any other allocation.
\end{definition}

Throughout this paper, we focus on exploring the fairness guarantees of the MNW rule.
Therefore, we formally define it as follows.
\begin{definition}[Maximum Nash Welfare Allocations]
The Nash welfare of an allocation $\bX$ is defined as $\emph{NSW}(\bX) = \prod_{i \in N} v_i(X_i)$.\footnote{Formally speaking, the Nash welfare is given by $(\prod_{i \in N} v_i(X_i))^{1/n}$. However, since we do not work with approximations, the two definitions are equivalent.} The MNW allocation maximizes the Nash welfare among all feasible allocations.
\end{definition}
Note that every MNW allocation is inherently PO, as any Pareto improvement would strictly increase the Nash welfare. 
Thus, in the remainder of this paper, all fairness guarantees for MNW allocations are satisfied simultaneously with PO.

\section{Fairness Guarantee for EFX and GMMS}

In this section, we study the fairness guarantee of MNW allocations for EFX and GMMS, and provide tight characterizations for the approximation ratios across different bi-valued settings.

\begin{definition}[Bi-valued Instances] 
An instance is called bi-valued if there exist two constants $a > b \geq 0$ such that for any agent $i \in N$ and item $g \in M$, we have $v_i(g) \in \{a, b\}$.
\end{definition}
We remark that when $b = 0$, such instances reduce to binary instances, for which any MNW allocation is EFX and GMMS~\cite{journals/tcs/AmanatidisBFHV21, conf/aaai/BarmanBMN18}.
Hence, in the following, we assume w.l.o.g. that  $b > 0$.
Additionally, we rescale all the valuation functions such that $v_i(g) \in \{1, k\}$, where $k > 1$.
When $k$ is an integer, we refer to such special bi-valued instances as factored bi-valued instances. 

\begin{definition}[Personalized Bi-valued Instances] 
An instance is called personalized bi-valued if for any $i\in N$, there exist constants $a_i > b_i \geq 0$ such that for all items $g \in M$, we have $v_i(g) \in \{a_i, b_i\}$.
\end{definition}

Similarly, we normalize the valuation functions such that $v_i(g) \in \{1, k_i\}$ for every agent $i \in N$ and item $g \in M$, with $k_i > 1$\footnote{This normalization assumes $b_i > 0$ for all $i \in N$. The boundary case where $b_i = 0$ for some agents is discussed in Appendix~\ref{section: discussion}.}.

Throughout this paper, we use $\bX$ to denote an arbitrarily fixed MNW allocation and study its fairness guarantees.
We also assume $m\geq n$ throughout this paper, as otherwise all fairness notions are trivial to guarantee.

Given valuation function~$v$, we call an item $g$ \emph{large} to agent $i$ if $v_i(g)=k_i$, and \emph{small} to agent $i$ if $v_i(g)=1$.  
For each agent $i$, let $L_i = \{g \in M:v_i(g) = k_i\}$ and $S_i = \{g \in M: v_i(g) = 1\}$ denote the set of large and small items for $i$, respectively.

Before we proceed, we first show some simple yet powerful properties of the MNW allocation $\bX$, which will be useful to show the approximation guarantee for EFX, GMMS, and PMMS.

\begin{lemma} \label{lemma:MNW_pb_dif}
    For all agents $i,j\in N$ and $g\in X_j$, we have
    \begin{equation*}
        \frac{v_i(g)}{v_i(X_i) + v_i(g)} \leq \frac{v_j(g)}{v_j(X_j)}.
    \end{equation*}    
\end{lemma}
\begin{proof}
    By the definition of MNW allocation, reallocating item $g$ from $X_j$ to $X_i$ does not increase the Nash welfare.
    Therefore we have 
    \begin{equation*}
        (v_i(X_i) + v_i(g)) \cdot (v_j(X_j) - v_j(g)) \leq v_i(X_i) \cdot v_j(X_j).
    \end{equation*}
    The lemma follows by rearranging the inequality.
\end{proof}

By a summation over all items $g\in X_j$, we obtain the following corollary for personalized bi-valued instances.

\begin{corollary}\label{corollary:MNW_sum_of_pb_dif}
    For all agents $i,j\in N$, we have 
    \begin{equation*}
        \frac{k_i\cdot |X_j\cap L_i|}{v_i(X_i) + k_i}
        + \frac{|X_j\cap S_i|}{v_i(X_i) + 1} \leq 1.
    \end{equation*}
\end{corollary}

Recall that for bi-valued instances, Amanatidis et al.~\cite{journals/tcs/AmanatidisBFHV21} showed that every MNW allocation is EFX.
Unfortunately, the property ceases to hold for personalized bi-valued instances, as independently observed by Jin and Tao~\cite{journals/corr/abs-2507-18251} and Byrka et al.~\cite{journals/corr/abs-2507-14957} in their recent works.
However, the approximation guarantee of MNW to EFX has not been studied yet.

In Section~\ref{ssec:pb_EFX}, we show that for personalized bi-valued instances, every MNW allocation guarantees $2/3$-EFX, and the approximation ratio is tight.

We further investigate the fairness guarantee of MNW allocations for GMMS.
We show that every MNW allocation guarantees $2/3$-GMMS for personalized bi-valued instances, and the approximation ratio is tight even when the instances are factored bi-valued.
Interestingly, we show that for personalized bi-valued instances, EFX implies $2/3$-GMMS (in Section~\ref{ssec:EFX_imply_2/3_MMS}), which might be of independent interest.
This readily implies the $2/3$-GMMS guarantee for (factored) bi-valued instances, as the MNW allocations are EFX.
While the MNW allocations do not guarantee EFX for personalized bi-valued instances, we show that they continue to preserve the $2/3$ approximation ratio guarantee for GMMS (in Section~\ref{ssec:pb_MMS}).

\subsection{EFX Guarantee for Personalized Bi-valued Instances}\label{ssec:pb_EFX}

In this section, we characterize the approximation guarantee of MNW allocation to EFX for personalized bi-valued instances, and prove the following main result.

\begin{theorem} \label{theorem:pb_mnwefx}
    For personalized bi-valued instances, every MNW allocation is $2/3$-EFX, and the approximation guarantee is tight.
\end{theorem}

We first provide the following hard instance showing that the MNW allocation does not guarantee an approximation ratio better than $2/3$ for EFX.

\begin{example} \label{example:MNW_not_EFX}
Consider an instance with $n = 2$ agents and a set of items $M = \{g_1, \dots, g_{m}\}$ with valuations defined as follows.
Note that we have $k_1 = m$ and $k_2 = 2m$.
It can be verified that the MNW allocation allocates item $g_2$ to agent $2$ and all other items to agent $1$ (as circled in Table~\ref{table:mnwnotefx}).
\begin{table}[h]
    \centering
    \caption{Hard instance for the approximation of EFX by MNW} \label{table:mnwnotefx}
    \renewcommand{\arraystretch}{1.15}
    \newcommand{\circval}[1]{\tikz[baseline=-0.6ex]{\node[draw,circle,inner sep=2.5pt] {$#1$};}}
    \setlength{\tabcolsep}{10pt}
    \begin{tabular}{lcccccc}
    \toprule
                & $g_1$ & $g_2$ & $g_3$ & $g_4$ & $\cdots$ & $g_m$ \\
    \midrule
    Agent 1       & \circval{m} & $1$ &\circval{1} & \circval{1} & $\cdots$ & \circval{1} \\
    Agent 2       & $2m$ & \circval{2m} & $1$ & $1$ & $\cdots$ & $1$ \\
    \bottomrule
    \end{tabular}
\end{table}

For agent 2, the EFX approximation ratio is 
{
    \fontsize{9}{10}\selectfont
$$
\frac{v_2(X_2)}{v_2(X_1-g_3)} = \frac{2m}{3m - 3} \xrightarrow[m\to\infty]{} \frac{2}{3}.
$$
}
\end{example}

Hence, we have established an upper bound of $2/3$ for the approximation ratio.
Next, we move to prove the matching lower bound.
Assume for the sake of contradiction that there exist agents $i$ and $j$, and $g\in X_j$ such that
\begin{equation}
    v_i(X_i) < \frac{2}{3}\cdot v_i(X_j - g).
    \label{equation:i_not_2/3_EFX_to_j}
\end{equation}

We first establish a few immediate observations.

Let $l_i = |X_i\cap L_i|$ and $s_i = |X_i\cap S_i|$ be the number of large and small items allocated to agent $i$, respectively.
Similarly, let $l_j = |X_j\cap L_i|$ and $s_j = |X_j\cap S_i|$.
Note that all these variables are defined with respect to $v_i$.
Since the MNW allocation $\bX$ is EF1, $X_j$ must contain both large and small items, i.e., $l_j \geq 1$ and $s_j \geq 1$ (as otherwise EF1 implies EFX).
Furthermore, the property of EF1 implies $v_i(X_j) \leq v_i(X_i) + k_i$, which further implies $v_i(X_i) < 2k_i$ because otherwise
\begin{equation*}
    v_i(X_i) \geq v_i(X_j) - k_i \geq v_i(X_j) - \frac{1}{2}\cdot v_i(X_i),
\end{equation*}
which contradicts with $i$ not being $2/3$-EFX towards $j$.
Consequently, we also have $v_i(X_j) < 3k_i$ and $l_j \in \{1,2\}$.
Next, we provide upper bounds on $v_j(X_j)$ and $v_j(X_i)$ in terms of $v_i(X_i)$, which will be the building blocks for further analysis.

\begin{lemma} \label{lemma:v_j(X_j)<=v_i(X_i)+1}
    We have $v_j(X_j) \leq v_i(X_i) + 1$.
\end{lemma}
\begin{proof}
    By Lemma~\ref{lemma:MNW_pb_dif}, if there exists $g\in X_j$ such that $v_j(g) = 1$, then we have
    \begin{equation*}
        v_j(X_j) \leq v_j(g)\cdot \frac{v_i(X_i) + v_i(g)}{v_i(g)}
        = \frac{v_i(X_i)}{v_i(g)} + 1 \leq v_i(X_i) + 1,
    \end{equation*}
    where the equality uses $v_j(g) = 1$ and the second inequality uses $v_i(g) \geq 1$.

    Hence, it suffices to argue that there exists an item in $X_j$ that is small to agent $j$.
    Assume otherwise, i.e., all items in $X_j$ are large to agent $j$.
    Then we have $v_j(X_j) = |X_j|\cdot k_j$.
    Since there exists an item $g\in X_j$ with $v_i(g) = k_i$, we obtain
    \begin{equation*}
        v_j(X_j) \leq v_j(g)\cdot \frac{v_i(X_i) + v_i(g)}{v_i(g)}
        = v_j(g)\cdot \frac{v_i(X_i) + k_i}{k_i}
        < 3\cdot v_j(g),
    \end{equation*}
    where the second inequality holds since $v_i(X_i) < 2k_i$.
    Therefore we have $|X_j| \leq 2$, which implies $l_j=s_j=1$.
    Since $i$ is not $2/3$-EFX towards $j$, we have $v_i(X_i) < k_i$.
    Now consider a new allocation obtained by reallocating the item $g$ that is large to agent $i$ from $X_j$ to $X_i$.
    The Nash welfare of the new allocation is strictly larger than that of $\bX$ because
    \begin{equation*}
        (v_i(X_i) + k_i)\cdot k_j > v_i(X_i)\cdot (2k_j),
    \end{equation*}
    which is a contradiction.
    Consequently, there must exist an item in $X_j$ that is small for agent $j$, thereby completing the proof.
\end{proof}

\begin{lemma}
    We have $v_j(X_i) < v_i(X_i)$.
\end{lemma}
\begin{proof}
    Suppose $v_j(X_i) \geq v_i(X_i)$, then we can strictly increase the Nash welfare by exchanging bundles $X_i$ and $X_j$ because
    \begin{equation*}
        v_i(X_j)\cdot v_j(X_i) > (v_i(X_i) + 1)\cdot v_i(X_i) \geq v_j(X_j)\cdot v_i(X_i),
    \end{equation*}
    where the second inequality holds by Lemma~\ref{lemma:v_j(X_j)<=v_i(X_i)+1}.
\end{proof}

As an immediate corollary, we conclude that $X_i$ contains at least one large item, i.e., $l_i\geq 1$, because otherwise we have $v_i(X_i) = |X_i|$, which leads to a contradiction that $v_j(X_i) < |X_i|$.
Note that since $v_i(X_i) < 2k_i$, we have $l_i = 1$ and $v_i(X_i) = k_i + s_i$.

Next, we derive a contradiction for each of the two cases for $l_j\in \{1,2\}$.

\paragraph{When $l_j=1$.}
Intuitively, since $X_j$ contains only one item that is large to agent $i$, we can show that the bundle value $v_i(X_j)$ is upper bounded, which leads to a contradiction with $i$ not being $2/3$-EFX towards $j$.
To establish the upper bound, we show that there are not too many small items in $X_j$, using the property of MNW allocation.
By Corollary~\ref{corollary:MNW_sum_of_pb_dif}, we have 
\begin{equation*}
    s_j \leq (v_i(X_i) + 1) \cdot \left( 1- \frac{k_i}{v_i(X_i) + k_i} \right) = (v_i(X_i) + 1) \cdot \frac{v_i(X_i)}{v_i(X_i) + k_i},
\end{equation*}
which implies
\begin{align*}
    \frac{v_i(X_i)}{v_i(X_j) - 1} = \frac{v_i(X_i)}{k_i + s_j - 1} & \geq \frac{v_i(X_i)}{k_i - 1 + (v_i(X_i) + 1) \cdot \frac{v_i(X_i)}{v_i(X_i) + k_i}} \\
    & = \frac{1}{\frac{k_i-1}{v_i(X_i)} + \frac{v_i(X_i) + 1}{v_i(X_i) + k_i}} \\
    & = \frac{1}{1 + \frac{k_i(k_i-1)}{v_i(X_i)(v_i(X_i)+k_i)}} \\
    & \geq \frac{1}{1 + \frac{k_i(k_i-1)}{k_i\cdot 2 k_i}} 
    = \frac{2k_i}{3k_i-1} \geq \frac{2}{3},
\end{align*}
where the second inequality holds since $v_i(X_i) \geq k_i$.
Hence, we have a contradiction.

\paragraph{When $l_j=2$.}

In this case, bundle $X_j$ contains two items that are large to agent $i$, i.e., having value $k_i$.
On the other hand, Lemma~\ref{lemma:v_j(X_j)<=v_i(X_i)+1} implies that these two large items have limited values to agent $j$.
Therefore, intuitively speaking, we should be able to swap some of these large items in $X_j$ with some small items in $X_i$ to strictly increase the Nash welfare, leading to a contradiction.
We first show that these two items are also large to agent $j$.
Let $g\in X_j$ be an item with $v_i(g) = k_i$.
By Lemma~\ref{lemma:MNW_pb_dif}, we have
\begin{equation*}
    v_j(g) \geq \frac{k_i}{v_i(X_i) + k_i}\cdot v_j(X_j) > \frac{1}{3}\cdot v_j(X_j) \geq \frac{1}{3}\cdot |X_j| \geq 1,
\end{equation*}
where the last inequality holds by $s_j \geq 1$ and $l_j = 2$.

Since the above inequality is strict, we have $v_j(g) = k_j$, i.e., $g$ is also large to agent $j$.
Note that the above inequality also implies $v_j(X_j) < 3\cdot v_j(g) = 3k_j$.
Therefore, except for the two items in $X_j$ that are large to both agents $i$ and $j$, all other items in $X_j$ are small to $j$.
Now, consider a new allocation obtained by swapping a large item in $X_j$ with all items in $X_i$ that are small to agent $i$.
Recall that $l_i = 1$ and $v_i(X_i) = k_i +s_i$.
The increase in Nash welfare after swapping is at least
\begin{align*}
    2k_i\cdot (k_j+s_j+s_i) - (k_i+s_i)\cdot (2k_j + s_j)
    & = 2k_i\cdot (s_j+s_i) - k_i\cdot s_j - 2s_i\cdot k_j - s_i\cdot s_j \\
    & = s_j\cdot (k_i - s_i) + 2s_i\cdot (k_i-k_j) > 0,
\end{align*}
where the inequality holds because $s_j\geq 1$, $s_i < k_i$ (implied by $v_i(X_i) = k_i + s_i < 2k_i$) and $k_i > k_j$ (implied by $v_j(X_j) = 2k_j + s_j \leq v_i(X_i) + 1 < 2k_i + 1$).
Hence, we have a contradiction.

Consequently, for any $i, j \in N$, $g \in X_j$, we must have $v_i(X_i) \geq \frac{2}{3} \cdot v_i(X_j - g)$. 
This completes the proof of the lower bound.

\subsection{Implication of EFX to Approximate GMMS}\label{ssec:EFX_imply_2/3_MMS}

In this section, we show that for personalized bi-valued instances, every EFX allocation is also $2/3$-GMMS, which improves upon the approximation guarantee for the general additive valuations (where EFX implies $10/17$-GMMS~\cite{journals/corr/abs-2608-30267}).
Furthermore, by combining our result with existing polynomial-time algorithms for EFX (Jin and Tao~\cite{journals/corr/abs-2507-18251} and Byrka et al.~\cite{journals/corr/abs-2507-14957}), we demonstrate that EFX and $2/3$-GMMS allocations can be computed efficiently for personalized bi-valued instances.

Notably, our analysis will be reused in the next section to show that the MNW allocation also guarantees a $2/3$-approximation of GMMS for personalized bi-valued instances.

\begin{theorem} \label{theorem:pb_efx_imply_2/3_MMS}
    For personalized bi-valued instances, every EFX allocation is $2/3$-GMMS, and the approximation guarantee is tight.
\end{theorem}
To establish Theorem~\ref{theorem:pb_efx_imply_2/3_MMS}, it suffices to prove that every EFX allocation is $2/3$-MMS.
By definition, GMMS requires the MMS guarantee for every subset of agents. 
Since EFX holds for any subset of agents, the same $2/3$-MMS argument applies to the instance induced by any subset of agents, which implies $2/3$-GMMS.

On the other hand, since $\alpha$-GMMS implies $\alpha$-MMS, any upper bound for approximating MMS under EFX immediately applies to GMMS as well.
Therefore, in the remainder of the proof, we focus on showing that every EFX allocation is $2/3$-MMS, and that the guarantee is tight.
We first show an instance for which there exists an MNW allocation (which is also EFX) that only guarantees a $2/3$ approximation for MMS.

\begin{example} \label{example:MNW/EFX_allocation_not_MMS}
Consider the instance with $n=3$ agents and $m=6$ goods shown in Table~\ref{table:hardness_23}, where item values are either $1$ or $2$.
Note that the instance is factored bi-valued.
It can be verified that the allocation (indicated by the circles) maximizes Nash welfare and is EFX.
However, the value for agent $1$ is $2$, while her maximin share is $3$.
\begin{table}[htbp] 
    \caption{Hard instance for the approximation of GMMS by EFX} 
    \label{table:hardness_23}
    \renewcommand{\arraystretch}{1.15}
    \newcommand{\circval}[1]{\tikz[baseline=-0.6ex]{\node[draw,circle,inner sep=2.5pt] {$#1$};}}
    \setlength{\tabcolsep}{10pt}
    \begin{center}
    \begin{tabular}{lcccccc}
    \toprule
            & $g_1$ & $g_2$ & $g_3$ & $g_4$ & $g_5$ & $g_6$ \\
    \midrule
    Agent 1 & \circval{2} & $2$ & $2$ & $1$ & $1$ & $1$ \\
    Agent 2 & $2$ & \circval{2} & \circval{2} & $1$ & $1$ & $1$ \\
    Agent 3 & $2$ & $2$ & $2$ & \circval{2} & \circval{2} & \circval{2} \\
    \bottomrule
    \end{tabular}
    \end{center}
\end{table}
\end{example}

In the following, we show that any EFX allocation $\bX$ guarantees at least $2/3$-MMS for personalized bi-valued instances. 
Throughout the analysis, we highlight the places where the property of EFX is used to establish some statements.
Then in Section~\ref{ssec:pb_MMS}, we show that the same statements can be proved using the properties of MNW allocations. 
Therefore, we can reuse the same analysis to derive the $2/3$-MMS guarantee for MNW allocations in the next section.

We fix an arbitrary agent $i$ and show that $v_i(X_i) \geq \frac{2}{3}\cdot \MMS_i$.
For convenience of notation, in the remainder of Section~\ref{ssec:EFX_imply_2/3_MMS}, we omit the subscript in $v_i$, $k_i$, and $\MMS_i$.
Furthermore, we call an item large (resp. small) if it is large (resp. small) to agent $i$.
Therefore, we also omit the subscript in $L_i$ and $S_i$, and use $l = |L_i|$ and $s=|S_i|$ to denote the number of large and small items, respectively.
We keep the subscript in $X_i$ as sometimes we will refer to bundles allocated to other agents, e.g., $X_j$ for some agent $j\neq i$.
Note that we have $X_j \neq \emptyset$ for all $j\in N$, as otherwise (by property of EF1) we have $m < n$.

We first show two simple lower bounds for $v(X_i)$ using the property of EF1.

\begin{lemma}\label{lemma:value_bound}
We have $v(X_i) > \MMS - k$.
\end{lemma}
\begin{proof}
Since $\bX$ is EF1, we have $v(X_j) \leq v(X_i) + k$ for any other agent $j \in N$.
Since the maximin share is at most the proportionality, we have 
$$
\MMS
\leq \frac{1}{n}\cdot \sum_{j\in[n]} v(X_j)
< \frac{1}{n}\cdot n \cdot (v(X_i)+k)
= v(X_i)+k.
$$

Rearranging the inequality gives the lemma.
\end{proof}

Note that when $\MMS\geq 3k$, the above lemma readily gives a lower bound of $2/3$ for the approximation ratio of MMS.
Hence, in the following analysis, we assume that $\MMS < 3k$.
Next, we show that $v(X_i)$ can also be lower bounded based on the number of large items.

\begin{lemma} \label{lemma:valuelowerbound}
    We have $v(X_i) \geq k\cdot \left\lfloor \frac{l}{n} \right\rfloor$.
\end{lemma}
\begin{proof}
    Assume for the sake of contradiction that $v(X_i) < k\cdot \floor{\frac{l}{n}}$, which implies that $X_i$ contains at most $\floor{\frac{l}{n}}-1$ large items.
    Then by the pigeonhole principle, there exists another bundle $X_j$ containing at least $\floor{\frac{l}{n}}+1$ large items (under $v$).
    Hence, even after removing a large item from $X_j$, its value under $v$ is at least $k\cdot \floor{\frac{l}{n}}$, strictly larger than $v(X_i)$, leading to a contradiction with the allocation being EF1.
\end{proof}

By Lemma~\ref{lemma:valuelowerbound}, when $\left\lfloor \frac{l}{n} \right\rfloor \geq 2$, we have $v(X_i) \geq 2k > \frac{2}{3}\cdot \MMS_i$. 
Therefore, it suffices to consider the case that $v(X_i) < 2k$ and $\left\lfloor \frac{l}{n} \right\rfloor \leq 1$, i.e., $l<2n$.
Next, we rule out the simple case of $v(X_i) < k$.
Note that the proof of the following lemma relies on the property of EFX.

\begin{lemma} \label{lemma:vi(Xi)<k}
    If $v(X_i) < k$, then we have $v(X_i) \geq \MMS$.
\end{lemma}
\begin{proof}
Since $v(X_i)<k$, bundle $X_i$ contains only small items. 
Therefore, we have $v(X_i) = |X_i|$.
Furthermore, by the property of EFX, every other bundle $X_j$ either contains a single large item, or contains no large item and at most $v(X_i)+1$ small items.

Note that there are $n-l$ bundles (including $X_i$) containing only small items (recall that $l$ is the total number of large items).  
Since each of these bundles contains at most $v(X_i)+1$ small items, we have $v(X_i) \geq \left\lfloor \frac{s}{n-l} \right\rfloor$.
On the other hand, notice that in any $n$-partition of the items, there are at least $n-l$ bundles containing only small items.
The minimum value among these bundles is therefore at most $\left\lfloor \frac{s}{n-l} \right\rfloor$.
Therefore we have $v(X_i) \geq \MMS$, as claimed.
\end{proof}

In summary, in the following, it remains to consider the situation with
\begin{itemize}
    \item limited number of large items: $l < 2n$,
    \item upper bounded maximin share: $\MMS < 3k$, and
    \item upper and lower bounded $v(X_i)$: $k\leq v(X_i) < 2k$.
\end{itemize}

Since $v(X_i) < 2k$, the property of EF1 implies that every bundle contains at most $2$ large items.
Therefore, we can partition the bundles into groups depending on the number of large items they contain as follows:
\begin{itemize}
    \item Small-item-only bundles: $N_0 = \{j \in [n]: X_j \subseteq S \}$.
    \item One-large-item bundles: $N_1 = \{j\in [n]: |X_j\cap L|=1 \}$.
    \item Two-large-item bundles: $N_2 = \{j \in [n]: |X_j\cap L|=2 \}$. The EFX property implies that these bundles contain no small items.
\end{itemize}

Note that we have $l = |N_1| + 2|N_2|$ and $n = |N_0| + |N_1| + |N_2|$, which gives
\begin{equation*}
    2n-l = 2\cdot |N_0| + |N_1|.
\end{equation*}

Next, we introduce a crucial lemma for upper-bounding the value of $\MMS$.
Intuitively, to help understand the lemma, one can interpret $2n-l$ as the shortage of large items towards securing two large items for each bundle, and these shortages should be filled by small items.
Since there are only $s$ small items, there must be some shortage that receives at most $\left\lfloor \frac{s}{2n - l}\right\rfloor$ small items, resulting in an upper bound for $\MMS$.

\begin{lemma}\label{lemma:case4MMS}
Let $t =\left\lfloor \frac{s}{2n - l}\right\rfloor$. We have
$\MMS \le \max\{k+t, 2t+1\}$.
\end{lemma}
\begin{proof}
Assume for the sake of contradiction that $\MMS > \max\{k+t, 2t+1\}$.
Let $(P_1,\dots,P_n)$ be an MMS partition for agent $i$, which satisfies $v(P_j) > \max\{k+t, 2t+1\}$ for all $j\in [n]$.
We further let $x_0 := |\{j\in[n] : P_j \cap L = \emptyset\}|$ denote the number of bundles with no large items, and $x_1 := |\{j\in[n] : |P_j \cap L| = 1 \}|$ denote the number of bundles with one large item.
Note that we have $2x_0+x_1 \ge 2n-l$ as we need at least $2n-l$ extra large items to ensure that each bundle contains at least two large items.

Since $v(P_j) > \max\{k+t, 2t+1\}$ for all $j\in [n]$, we conclude that
\begin{itemize}
    \item each bundle $P_j$ with no large items contains at least $2t + 2$ small items.
    \item each bundle $P_j$ with one large item contains at least $t + 1$ small items.
\end{itemize}

Then, counting the number of small items in these bundles, we get
$$
s \ge x_0\cdot (2t+2) + x_1\cdot (t+1)
= (2x_0 + x_1)\cdot (t+1) \ge (2n-l)(t+1),
$$
which leads to $t+1 \leq \frac{s}{2n-l}$ and
contradicts with $t=\left\lfloor \frac{s}{2n - l}\right\rfloor$.
\end{proof}

In the following, we show that $\frac{3}{2}\cdot v(X_i) \geq \max\{k+t, 2t+1\}$, which implies $v(X_i) \geq \frac{2}{3}\cdot \MMS$ by Lemma~\ref{lemma:case4MMS}.
Intuitively speaking, we give lower bounds in terms of $v(X_i)$ for the number of small items in each bundle other than $X_i$, which allows us to lower bound $s$ (the total number of small items) in terms of $v(X_i)$.
This then translates into upper bounds for $t =\left\lfloor \frac{s}{2n - l}\right\rfloor$, and enables us to establish $\frac{3}{2}\cdot v(X_i) \geq \max\{k+t, 2t+1\}$.

Recall that $N_0$ is the set of bundles containing only small items and $N_1$ is the set of bundles containing one large item.
By the property of EFX, we have the following upper bound on the number of small items in these bundles:
\begin{itemize}
    \item for all $j\in N_0$, we have $|X_j\cap S| \leq \alpha:=\lfloor v(X_i)+1\rfloor$.
    \item for all $j\in N_1$, we have $|X_j\cap S| \leq \beta:=\lfloor v(X_i)-k+1\rfloor$.
\end{itemize}

Note that to establish the first upper bound, we only need EF1. 
Furthermore, since $i\in N_0\cup N_1$, the above upper bounds also hold for agent $i$.
Moreover, for agent $i$, the corresponding inequality is strict.
Therefore, we establish the following upper bound on the total number of small items.

\begin{lemma} \label{lemma:upper_bound_s_by_alpha_beta}
    We have $s < (2n-l)\cdot \max\{\alpha /2, \beta\}$.
\end{lemma}
\begin{proof}
    Summing the upper bounds on the number of small items in each bundle, we obtain
    \begin{align*}
        s  & <  |N_0|\cdot \alpha + |N_1|\cdot \beta \\
        & \leq (2|N_0| + |N_1|) \cdot \max\{\alpha / 2, \beta\} \\
        & = (2n-l)\cdot \max\{\alpha / 2, \beta\},
    \end{align*}
    where the first inequality is strict because $i\in N_0\cup N_1$ and the corresponding upper bound ($\alpha$ or $\beta$) is strict.
\end{proof}

Finally, we show the lower bound for $v(X_i)$.

\begin{lemma}\label{lemma:lower_bound_v(X_i)_by_max{k+t,2t+1}}
    We have $\frac{3}{2}\cdot v(X_i) \geq \max\{ k+t,2t+1 \}$.
\end{lemma}
\begin{proof}
    We discuss two cases depending on the value of $\max\{\alpha / 2, \beta\}$.

    If $\beta \geq \alpha/2$, then by Lemma~\ref{lemma:upper_bound_s_by_alpha_beta}, we have $s < (2n-l)\cdot \beta$.
    Since both $t$ and $\beta$ are integers, we have
    \begin{equation*}
        t = \left\lfloor \frac{s}{2n-l} \right\rfloor \leq \beta-1
        = \lfloor v(X_i)-k+1\rfloor - 1
        \leq v(X_i) -k.
    \end{equation*}
    By $t \leq v(X_i) - k$, we have 
    $$
    \frac{3}{2}\cdot v(X_i) > v(X_i) = k+(v(X_i) - k) \geq k+t.
    $$
    Furthermore, using the fact that $v(X_i) < 2k$, we derive the following bounds:
    $$\frac{3}{2} \cdot v(X_i) > 2(v(X_i)-k) + 1 \geq 2t+1.
    $$
    
    Therefore, we have $\frac{3}{2}\cdot v(X_i) \geq \max\{ k+t,2t+1 \}$.

     If $\beta < \alpha/2$, then Lemma~\ref{lemma:upper_bound_s_by_alpha_beta} gives $s < (2n-l)\cdot \alpha/2$, which implies $\alpha > \frac{2s}{2n-l}$.
    Again, since both $t$ and $\alpha$ are integers, we obtain
    \begin{equation*}
        2t+1 = 2\cdot\left\lfloor \frac{s}{2n-l} \right\rfloor +1 \leq \alpha 
        = \lfloor v(X_i)+1\rfloor
        \leq v(X_i)+1.
    \end{equation*}

    By the above inequality, we have 
    $\frac{3}{2}\cdot v(X_i) \geq k+t$ because $v(X_i) \geq k$ and $2t\leq v(X_i)$.

    On the other hand,
    $\frac{3}{2}\cdot v(X_i) \geq 2t+1$ when $v(X_i) \geq 2$. 
    If $v(X_i) < 2$, as $k \leq v(X_i) <2k$, it must be the case that $X_i$ contains a single item and $v(X_i) = k\in (1,2)$.
    In this case we have $\alpha = 2$ and $t=0$, and $\frac{3}{2}\cdot v(X_i) \geq 2t+1$ trivially holds.
    Therefore we also have $\frac{3}{2}\cdot v(X_i) \geq \max\{ k+t,2t+1 \}$, which completes the proof.
\end{proof}

Therefore, given any EFX allocation $\bX$, we have $v(X_i) \geq \frac{2}{3} \cdot \MMS$.

\subsection{GMMS Guarantee for Personalized Bi-valued Instances}\label{ssec:pb_MMS}

In this section, we show that for personalized bi-valued instances, the MNW allocation is $2/3$-GMMS.
As we have shown in Section~\ref{ssec:pb_EFX}, the MNW allocation does not guarantee EFX.
Thus, we cannot simply use the result from the previous section.
However, we show that the same analysis can be applied with some minor changes to derive the same approximation guarantee.

\begin{theorem} \label{theorem:pb_mnwmms}
    For personalized bi-valued instances, every MNW allocation is $2/3$-GMMS, and the approximation guarantee is tight.
\end{theorem}
Note that for any subset of agents $N' \subseteq N$, the induced allocation on the induced instance is still an MNW allocation. 
Therefore, it suffices to prove that every MNW allocation is $2/3$-MMS.
In the following, we fix an arbitrary agent $i$ and show that $v_i(X_i) \geq \frac{2}{3}\cdot \MMS_i$.

We observe that the analysis from the previous section mainly uses the property of EF1, which is guaranteed by MNW allocations, as shown by Caragiannis et al.~\cite{journals/teco/CaragiannisKMPS19}.
Hence, there are only three places where the property of EFX is used: (1) Lemma~\ref{lemma:vi(Xi)<k}, (2) bundles in $N_2$ contain no small items when $k \leq v_i(X_i) < 2k$ and (3) when establishing the upper bound of $\beta$ on the number of small items in bundles in $N_1$.
In the following, we fix any MNW allocation $\bX$, and show that the same statements can be established by the property of the MNW allocation.

Unless otherwise specified, we refer to an item as large (resp. small) if it is large (resp. small) to agent $i$.
We use $L$ and $S$ to denote the set of large and small items, $l=|L|$, $s=|S|$ and $k=k_i$.

We first introduce the following lemma that is analogous to Lemma~\ref{lemma:vi(Xi)<k}.

\begin{lemma}\label{lemma:vi(Xi)<k_alternative}
If $v_i(X_i) < k$, then $v_i(X_i) \geq \MMS_i$.
\end{lemma}
\begin{proof}
Since $v_i(X_i) < k$, all items in $X_i$ are small. 
By the EF1 property, every other bundle contains at most one large item. 
In the following, we show that every other bundle $X_j$ either contains a single large item, or contains no large item and at most $v_i(X_i)+1$ small items.

Consider any bundle $X_j$ that contains an item $g$ with $v_i(g) = k$.
Suppose $X_j$ also contains other items (which are small to agent $i$).
By Lemma~\ref{lemma:MNW_pb_dif},
\begin{equation*}
    v_j(g) \geq \frac{k}{v_i(X_i) + k}\cdot v_j(X_j) > \frac{1}{2}\cdot v_j(X_j) \geq \frac{1}{2}\cdot |X_j| \geq 1.
\end{equation*}
Therefore, we have $v_j(g) = k_j$.
Note that the above inequality also implies $v_j(X_j) < 2\cdot v_j(g) = 2k_j$.
Thus, except for item $g\in X_j$, the other items in $X_j$ are small to $j$.
Fix any $g' \in X_j\setminus \{g\}$.
By Lemma~\ref{lemma:MNW_pb_dif}, we have
\begin{equation*}
    v_j(X_j) \leq v_j(g')\cdot \frac{v_i(X_i) + v_i(g')}{v_i(g')}
    = \frac{v_i(X_i)}{v_i(g')} + 1 \leq v_i(X_i) + 1.
\end{equation*}

Now, consider a new allocation obtained by swapping the large item in $X_j$ with all items in $X_i$.
The Nash welfare of the new allocation is strictly larger because
\begin{equation*}
    k\cdot (|X_i| + |X_j\cap S|) 
    \geq k\cdot (v_i(X_i) + 1)
    > v_i(X_i)\cdot v_j(X_j).
\end{equation*}

Hence, we have a contradiction, which leads to the conclusion that every bundle that contains a large item does not have any other items.
%
Consequently, by the same argument as in Lemma~\ref{lemma:vi(Xi)<k}, we can show that $v_i(X_i) \geq \left\lfloor \frac{s}{n-l} \right\rfloor \geq \MMS_i$, where $s$ and $l$ are the total number of small and large items, respectively.
\end{proof}

For MNW allocation in personalized bi-valued instances, we can also show that bundles in $N_2$ cannot contain any small item when $k \leq v_i(X_i) < 2k$.
\begin{lemma}
    Bundles in $N_2$ contain no small items if $k \leq v_i(X_i) < 2k$.
\end{lemma}
\begin{proof}
    Suppose for the sake of contradiction that there exists some $j \in N_2$ such that $|X_j \cap L| = 2$ and $|X_j \cap S| \geq 1$.
    By definition of $N_2$, there exist two items $g_1, g_2 \in X_j$, such that $v_i(g_1) = v_i(g_2) = k$ and hence $|X_j| \geq 3$.
    
    By Lemma~\ref{lemma:MNW_pb_dif}, we have 
    $$\frac{v_j(g_1)}{v_j(X_j)} \geq \frac{v_i(g_1)}{v_i(X_i) + v_i(g_1)} > \frac{1}{3},
    $$
    which implies $v_j(g_1) = k_j$ as $|X_j| \geq 3$.
    A similar argument establishes that $v_j(g_2) = k_j$.
    Moreover, by $\frac{k_j}{v_j(X_j)} > \frac{1}{3}$, we have $v_j(X_j) < 3k_j$ and $v_i(g) = v_j(g) = 1$ for any $g \in X_j \setminus \{g_1, g_2\}$.
    Let $v_j(X_j) = 2k_j + s_j$. 
    We must have $v_j(X_j) - s_j \leq v_i(X_i)$, as otherwise shifting all small items from $X_j$ to $X_i$ would yield a change in Nash welfare of
    $$(v_i(X_i) + s_j)\cdot (v_j(X_j) - s_j) - v_i(X_i)\cdot v_j(X_j) = s_j(v_j(X_j) - v_i(X_i)) - s_j^2 > 0,
    $$
    which contradicts the optimality of $\bX$. 
    
    Hence, we can conclude that $2k_j \leq v_i(X_i) < 2k$.
    Now, let $x = \floor{v_i(X_i) - k}$ and $\epsilon = v_i(X_i) - x - k \geq 0$. 
    Consider a new allocation $\bX'$ obtained by moving $g_1$ from $X_j$ to $X_i$ and moving $x$ small items from $X_i$ to $X_j$. 
    The resulting change in Nash welfare is given by:
\begin{align*}
    v_i(X_i') \cdot v_j(X_j') - v_i(X_i) \cdot v_j(X_j) 
    & \geq (v_i(X_i)+k - x)(k_j + s_j + x) - v_i(X_i)(2k_j + s_j) \\
    & = (2k+\epsilon)(k_j+s_j+x) - (k+x+\epsilon)(2k_j+s_j)\\
    & = s_j(k - x) + 2x(k - k_j) + \epsilon(x - k_j) \\
    & > s_j(k - x) + \epsilon(x - k_j) \tag{by $k_j < k$}\\
    & \geq k - x + \epsilon(\frac{x - k - \epsilon}{2}) \tag{by $s_j \geq 1$ and $2k_j \leq v_i(X_i)$}\\
    & \geq \frac{k - x - \epsilon}{2} > 0,
\end{align*}
where the last two inequalities hold by $0 \leq \epsilon < 1$ and $k + x + \epsilon = v_i(X_i) < 2k$, respectively.
This contradicts the optimality of $\bX$. 
Thus, bundles in $N_2$ cannot contain small items.
\end{proof}

Next, we provide an alternative definition of $\beta$, and show that we can apply a similar analysis as in Lemmas~\ref{lemma:upper_bound_s_by_alpha_beta} and~\ref{lemma:lower_bound_v(X_i)_by_max{k+t,2t+1}} to the new definition of $\beta$.
For ease of notation, in the following we use $x$ to denote $v_i(X_i)$.
Note that $x\in [k,2k)$.
We re-define $\beta := \left \lfloor \frac{x(x+1)}{x + k} \right\rfloor$.
We first give an upper bound on the number of small items in bundles of $N_1$.

\begin{lemma}\label{lemma:A_1_upper_bound_alternative}
    For any agent $j \in N_1$, we have $|X_j \cap S| \leq \beta$.
\end{lemma}
\begin{proof}
    Since $X_j$ contains a large item, by Corollary~\ref{corollary:MNW_sum_of_pb_dif}, we have 
    $$
    \frac{|X_j \cap S|}{x+1} + \frac{k}{x+k}\leq 1,
    $$
    which gives $|X_j \cap S| \leq \frac{x(x+1)}{x+k}$.
    The lemma follows since $|X_j \cap S|$ is an integer.
\end{proof}

Furthermore, if $i\in N_1$ then we have $|X_i\cap S| = x-k < \beta$.
Given the above upper bounds, we can re-establish Lemma~\ref{lemma:upper_bound_s_by_alpha_beta} under the new definition of $\beta$ (and the same definition of $\alpha$).
Next, we provide an analogous version of Lemma~\ref{lemma:lower_bound_v(X_i)_by_max{k+t,2t+1}} (recall that $t =\left\lfloor \frac{s}{2n - l}\right\rfloor$).

\begin{lemma}\label{lemma:lower_bound_v(X_i)_by_max{k+t,2t+1}_alternative}
    We have $\frac{3}{2}\cdot v_i(X_i) \geq \max\{ k+t,2t+1 \}$.
\end{lemma}
\begin{proof}
    As before, we discuss two cases depending on the value of $\max\{\alpha / 2, \beta\}$.
    The proof for the case of $\beta < \alpha/2$ is identical and thus omitted.
    In the following, we focus on the case of $\beta \geq \alpha/2$.
    By Lemma~\ref{lemma:upper_bound_s_by_alpha_beta}, we have $s < (2n-l)\cdot \beta$ and $t \leq \beta - 1$ since both $t$ and $\beta$ are integers.

    We first show that $k+t \leq \frac{3}{2}\cdot x$.
    Since $t\leq \beta-1 = \left \lfloor \frac{x(x+1)}{x + k} \right\rfloor-1$, it suffices to show that
    \begin{equation*}
        \frac{x(x+1)}{x + k}\leq \frac{3}{2}\cdot x-k+1,
    \end{equation*}
    which is equivalent to
    \begin{equation*}
        \frac{1}{2}\cdot x^2 + \frac{k}{2}\cdot x \geq k(k-1).
    \end{equation*}
    The above inequality holds since $x \geq k$.

    Next, we show $2t+1 \leq \frac{3}{2}\cdot x$.
    Since $2t+1 \leq 2\beta-1 = 2\left \lfloor \frac{x(x+1)}{x + k} \right\rfloor-1$, it suffices to show that
    \begin{equation*}
        \frac{2x(x+1)}{x + k}-1 \leq \frac{3}{2}\cdot x,
    \end{equation*}
    which is equivalent to
    \begin{equation*}
        x^2 - (3k-2)x \leq 2k.
    \end{equation*}

    It can be verified that the LHS achieves its maximum when $x\in \{k,2k\}$.
    Therefore we have
    \begin{equation*}
        x^2 - (3k-2)x \leq \max\{ 2k-2k^2, 4k-2k^2 \} \leq 2k.
    \end{equation*}

    In summary, we have 
    $$\frac{3}{2}\cdot v_i(X_i) = \frac{3}{2}\cdot x \geq \max\{ k+t,2t+1 \}, 
    $$
    which completes the proof of this lemma.
\end{proof}

Hence, by replacing the corresponding parts in the proofs in Section~\ref{ssec:EFX_imply_2/3_MMS}, we can reproduce the same approximation guarantee of MMS for personalized bi-valued instances, which implies the corresponding guarantee of GMMS.


\section{Fairness Guarantee for PMMS}

In this section, we investigate the approximation guarantee of MNW allocations for PMMS.
While previous results for GMMS imply that MNW allocations are at least $2/3$-PMMS for personalized bi-valued instances, we establish stronger tight bounds of $6/7$ and $1$ for bi-valued and factored bi-valued instances, respectively.

Furthermore, we provide a matching upper bound for personalized bi-valued instances, demonstrating that the $2/3$-PMMS guarantee is tight for such instances.

\paragraph{Roadmap.}
Our analysis for all three settings begins by fixing a pair of agents $i$ and $j$, and then shows $v_i(X_i) \geq \alpha\cdot \mu$ with the corresponding approximation ratio $\alpha \in\{2/3, 6/7, 1\}$, where $\mu = \MMS_i(X_i\cup X_j, 2)$.
We first consider (non-personalized) bi-valued instances, for which the MNW allocation is EFX.
If all items in $X_j$ are large to agent $i$, then we show that $v_i(X_i) \geq \mu$, using the property of MNW allocation.
If $X_j$ contains items that are small to agent $i$, then the property of EFX implies that $v_i(X_i)$ and $v_i(X_j)$ differ by at most $1$, which allows us to show $v_i(X_i) \geq \mu$ for factored instances because all item values are integral.
However, for non-factored instances, this does not hold (see Example~\ref{example:mnwnotpmms}).
Instead, we show that we must have $v_i(X_i) \geq 3$, which implies a lower bound of $6/7$ on the approximation ratio.
The above analysis does not work for personalized bi-valued instances, even if all items in $X_j$ are large to agent $i$ (see Example~\ref{example:pb_mnwpmms}).
However, when combined with our previous analysis of GMMS, this example establishes the tightness of the guarantee.

\begin{table}[htbp]
\centering
\caption{Approximation ratios of PMMS under different settings and cases}
\label{tab:pmms_ratios}
\begin{tabular}{c|c|c|c}
\arrayrulecolor{ruleblue}\specialrule{1.2pt}{0pt}{0pt}\arrayrulecolor{black}
\rowcolor{shadegray}    & \textbf{Factored Bi-valued} & \textbf{Non-factored Bi-valued} & \textbf{Personalized Bi-valued} \\
    \bluerule
    \textbf{$X_j\cap S \neq \emptyset$}              & $1$ & $6/7$  & \multirow{2}{*}{$2/3$ (Theorem~\ref{theorem:pb_mnwpmms})} \\ 
    \cline{1-3}
    \textbf{$X_j\cap S = \emptyset$}  & \multicolumn{2}{c|}{$1$ (Lemma~\ref{lemma:all_large})} \\
\arrayrulecolor{ruleblue}\specialrule{1.2pt}{0pt}{0pt}\arrayrulecolor{black}
\end{tabular}
\end{table}

Recall that for bi-valued instances, there exists $k>1$ such that $v_i(g)\in\{1,k\}$ for all $i\in N$ and $g\in M$.
Note that for bi-valued instances, the MNW allocation is EFX.
In this section, we show that the MNW allocation is $1$-PMMS and $6/7$-PMMS for factored bi-valued instances and non-factored bi-valued instances, respectively.

Fix any two agents $i, j \in N$, let $T = X_i \cup X_j$ and $\mu = \MMS_i(T,2)$.
We show approximations of PMMS for agent $i$.
Note that it suffices to consider the case that $v_i(X_i) < v_i(X_j)$, as otherwise $v_i(X_i) \geq v_i(T)/2 \geq \mu$.
We further define:
\begin{itemize}
    \item the set of large items for $i$ in $T$ as $L = \{g\in T:v_i(g) = k\}$, and $l = |L|$.
    \item the set of small items for $i$ in $T$ as $S = \{g\in T:v_i(g) = 1\}$, and $s = |S|$.
\end{itemize}

Unless otherwise specified, we call an item large (resp. small) if it is large (resp. small) to agent $i$.
Note that we can assume w.l.o.g. that $l>0$ and $s>0$, as otherwise all items in $T$ have the same value to agent $i$ and the property of EFX implies that for any $g \in T$, we can guarantee
$$v_i(X_i) \geq v_i(g)\cdot \left\lfloor\frac{|T|}{2}\right\rfloor \geq \mu.
$$

We first consider the case when all items in $X_j$ are large.

\begin{lemma} \label{lemma:all_large}
    If all items in $X_j$ are large, then we have $v_i(X_i) \geq \mu$.
\end{lemma}
\begin{proof}
    Assume the contrary that $v_i(X_i) < \mu$.
    Since all items in $X_j$ are large, the $s$ small items are all contained in $X_i$.
    Furthermore, all items in $X_j$ are also large to agent $j$, because otherwise we can swap a large item $g\in X_j$ that is small to agent $j$ with a small item in $X_i$ to Pareto improve the allocation, which contradicts with the allocation being MNW.

    Suppose in the MMS partition that defines $\MMS_i(T,2)$, the bundle with smaller value contains $l'$ large items and $s'$ small items.
    Note that the other bundle contains $l-l'$ large items and $s-s'$ small items, and we have
    \begin{equation*}
        v_i(T) - \mu = (l-l')\cdot k+(s-s') \geq l'\cdot k + s' = \mu.
    \end{equation*}

    Consider the allocation $(X'_i,X'_j)$ obtained by swapping $s-s'$ small items in $X_i$ with $l'-|X_i\cap L|$ large items in $X_j$.
    Clearly, after reallocation, we have $v_i(X'_i) = \mu$.
    Furthermore, we have
    \begin{equation*}
        v_j(X'_j) \geq (s-s') + (|X_j\cap L| - l' + |X_i\cap L|)\cdot k = (l-l')\cdot k+(s-s') = v_i(T) - \mu.
    \end{equation*}
    
    Hence, the Nash welfare of the new allocation is at least
    \begin{align*}
        \mu\cdot (v_i(T) - \mu)
        & = \left( \frac{v_i(T)}{2} \right)^2 - \left( \frac{v_i(T)}{2}-\mu \right)^2, \\
        & > \left( \frac{v_i(T)}{2} \right)^2 - \left( \frac{v_i(T)}{2}-v_i(X_i) \right)^2, \\
        & = v_i(X_i)\cdot (v_i(T)-v_i(X_i)) = v_i(X_i)\cdot v_j(X_j).
    \end{align*}
    which leads to a contradiction.
\end{proof}

Next, we consider the case when $X_j$ contains small items.
Since the MNW allocation is EFX for bi-valued instances, we immediately have $v_i(X_j) \leq v_i(X_i) + 1$.

Then, for factored instances, we have 
$$
v_i(X_i) \geq \left\lfloor \frac{v_i(T)}{2} \right\rfloor \geq \mu,
$$ since all items have integer values.
Hence, we have the following.

\begin{theorem}\label{theorem:fb-mnwpmms}
    For any factored bi-valued instance, every MNW allocation is PMMS.
\end{theorem}

Crucially, for \emph{factored} bi-valued instances, this result can be strengthened to a computational guarantee.
By combining with a recent polynomial-time algorithm for computing MNW allocation for factored bi-valued instances~\cite{akrami2025maximizing}, our result implies that an allocation satisfying EFX, PMMS, $2/3$-GMMS, and PO can be computed efficiently in this setting.

Next, we consider non-factored instances, i.e., $k$ is not an integer. 
We first present an example showing that in this case the MNW allocation is at most $6/7$-PMMS (see Example~\ref{example:mnwnotpmms}).

\begin{example}\label{example:mnwnotpmms}
    Consider an instance where $n = 2$ and $m = 6$ as shown in Table~\ref{table:mnwnotpmms}. 
    It can be verified that the following allocation $\bX$ maximizes the Nash welfare: $X_1 = \{g_1, g_2\}$ and $X_2 = \{g_3, g_4, g_5, g_6\}$.
    However, we have $v_1(X_1) = 3 < \MMS_{1}(X_1 \cup X_2,2) = 3.5$, which implies that $\bX$ is $6/7$-PMMS.
\begin{table}[h]
    \centering
    \caption{Hard instance for the approximation of PMMS by MNW (bi-valued)} \label{table:mnwnotpmms}
    \renewcommand{\arraystretch}{1.15}
    \newcommand{\circval}[1]{\tikz[baseline=-0.6ex]{\node[draw,circle,inner sep=2.5pt] {$#1$};}}
    \setlength{\tabcolsep}{10pt}
    \begin{tabular}{lcccccc}
    \toprule
                & $g_1$ & $g_2$ & $g_3$ & $g_4$ & $g_5$ & $g_6$ \\
    \midrule
    Agent 1       & \circval{1.5} & \circval{1.5} &1 & 1 & 1 & 1\\
    Agent 2       & 1 & 1 & \circval{1} & \circval{1} & \circval{1} & \circval{1} \\
    \bottomrule
    \end{tabular}
\end{table}

\end{example}

The above example demonstrates the difference between factored and non-factored instances: from the perspective of agent $1$, while bundle $X_2$ contains small items, we do not have $v_1(X_1) \geq \mu$.\footnote{It suffices to consider the case where $X_2$ contains small items, as Lemma~\ref{lemma:all_large} remains applicable to bi-valued instances.}
However, we can still use the property of EFX to derive $v_i(X_j) \leq v_i(X_i) + 1$, which gives
\begin{equation*}
    \frac{v_i(X_i)}{\mu} \geq \frac{v_i(X_i)}{v_i(T)/2} \geq \frac{2v_i(X_i)}{2v_i(X_i) + 1}.
\end{equation*}

Observe that the above lower bound is at least $6/7$ when $v_i(X_i) \geq 3$.
In the following, we show that if $v_i(X_i) < 3$ then we have $v_i(X_i) \geq \mu$.

\begin{lemma} \label{lemma:smallPMMS}
    If $v_i(X_i) < 3$, then $v_i(X_i) \geq \mu$.
\end{lemma}
\begin{proof}
    Note that in this case $X_i$ contains at most two items and $X_j$ contains at most three items since $v_i(X_j) \leq v_i(X_i) + 1 < 4$.
    Depending on whether the items in $X_i$ are small or large, there are five cases, with the corresponding $v_i(X_i) \in \{ 1,2,k,k+1,2k \}$ (recall that $k$ is not an integer).
    
    Assume for the sake of contradiction that $v_i(X_i) < \mu$.
    Observe that if $|X_i|\geq |X_j|-1$ and $v_i(g) \geq v_i(g')$ for all $g\in X_i$ and $g'\in X_j$, then we have $v_i(X_i) \geq \mu$ because there must exist a bundle in the MMS partition of $\MMS_i(T,2)$ with size $|X_i|$, whose value is at most $v_i(X_i)$.
    This enables us to quickly rule out two cases: when $v_i(X_i)=1$ ($X_i$ contains a single small item), and when $v_i(X_i) = 2k$ ($X_i$ contains two large items).
    
    For the remaining three cases, we use $v_i(X_i) < \mu < v_i(X_j)$ to derive contradictions.
    Note that both $\mu$ and $v_i(X_j)$ take values in the form of $a\cdot k + b$ for some integers $a,b\geq 0$. 
    \begin{itemize}
        \item When $v_i(X_i) = 2$, i.e., $X_i$ contains two small items, we have $v_i(X_j) \leq 3$.
        Note that $X_j$ must contain at least one large item as $l>0$.
        Therefore we have $v_i(X_j) = k+1$ with $k \leq 2$.
        However, there is no feasible value (in the form of $a\cdot k + b$) for $\mu\in (2,k+1)$.
        
        \item When $v_i(X_i) = k$, i.e., $X_i$ contains a single large item, we have $v_i(X_j) \leq k+1$.
        To satisfy $\mu \in (k,k+1)$, we must have $\mu = \lceil k \rceil = \lfloor v_i(X_j) \rfloor$.
        However, it is clearly impossible to partition $T$ into two bundles with value at least $\lceil k \rceil$.

        \item When $v_i(X_i) = k+1$, i.e., $X_i$ contains a large item and a small item, we have $v_i(X_j)\leq k+2$, which implies $\mu \in \{3,2k\}$ and $k\in(1,2)$.
        Furthermore, $X_j$ contains one large item $g$ and two small items.
        Now, consider the allocation obtained by swapping the large item $g\in X_j$ with the small item in $X_i$.
        Let $x = v_j(X_j-g) \geq 2$.
        Then we have a contradiction as the Nash welfare of the new allocation is strictly larger:
        \begin{align*}
            2k\cdot (x + 1) & = \left( \frac{x + 1}{2} + k \right)^2 - \left( \frac{x + 1}{2} - k \right)^2 \\
            & > \left( \frac{x + 1}{2} + k \right)^2 - \left( \frac{x + 1}{2} - 1 \right)^2 \\
            & = (k+1)\cdot (x+k) \ge v_i(X_i)\cdot v_j(X_j),
        \end{align*}
        
        where the first inequality holds for all $x\geq 2 > k$.
    \end{itemize}

    Therefore, we have contradictions in all cases, which implies $v_i(X_i)\geq \mu$.
\end{proof}

In summary, we establish the following theorem.

\begin{theorem}\label{theorem:bi-mnwpmms}
    Given any bi-valued instance, the MNW allocation is $6/7$-PMMS, and the approximation guarantee is tight.
\end{theorem}

Finally, we consider the personalized bi-valued instances.
Recall that by Theorem~\ref{theorem:pb_mnwmms}, any MNW allocation is at least $2/3$-GMMS (and hence $2/3$-PMMS).
Therefore, it suffices to present an example showing that the MNW allocation may not guarantee an approximation ratio better than $2/3$ for personalized bi-valued instances (see Example~\ref{example:pb_mnwpmms}).
Note that the example demonstrates that Lemma~\ref{lemma:all_large} (showing that $v_i(X_i)\geq \mu$ if all items in $X_j$ are large) does not hold for personalized bi-valued instances.

\begin{example} \label{example:pb_mnwpmms}
Consider the following instance with $n = 2$ agents and $m = 4$ goods $\{g_1, g_2, g_3, g_4\}$.
Each agent has additive valuations as shown in Table~\ref{table:pb_mnwpmms}. 
It can be verified that the following allocation $\bX$ maximizes the Nash welfare: $X_1 = \{g_3, g_4\}$ and $X_2 = \{g_1, g_2\}$.
However, we have $v_1(X_1) = 2 < \MMS_{1}(X_1 \cup X_2,2) = 3$, which implies that $\bX$ is $2/3$-PMMS.
\begin{table}[h]
    \centering
    \caption{Hard instance for the approximation of PMMS by MNW (personalized bi-valued)}\label{table:pb_mnwpmms}
    \renewcommand{\arraystretch}{1.15}
    \newcommand{\circval}[1]{\tikz[baseline=-0.6ex]{\node[draw,circle,inner sep=2.5pt] {$#1$};}}
    \setlength{\tabcolsep}{10pt}
    \begin{tabular}{lcccccc}
    \toprule
                & $g_1$ & $g_2$ & $g_3$ & $g_4$ \\
    \midrule
    Agent 1       & 2 & 2 & \circval{1} & \circval{1}\\
    Agent 2        & \circval{4} & \circval{4} & 1 & 1\\
    \bottomrule
    \end{tabular}
\end{table}
\end{example}

Hence, we arrive at the following theorem.

\begin{theorem} \label{theorem:pb_mnwpmms}
    For personalized bi-valued instances, every MNW allocation is $2/3$-PMMS, and the approximation guarantee is tight.
\end{theorem}

\section{Upper Bounds for Tri-valued Instances}
\label{section:tri_upper}

In this section, we show upper bounds on the approximation ratios of MNW allocations with respect to EFX, MMS, and PMMS for tri-valued instances.

\begin{definition}[Tri-valued Instances]
    An instance is called tri-valued if there exist constants $a > b > c \geq 0$ such that for any agent $i \in N$ and item $g \in M$, we have $v_i(g) \in \{a, b, c\}$.
\end{definition}

Then we demonstrate that for tri-valued instances, the upper bound for EFX matches the results for general instances established by Amanatidis et al.~\cite{journals/tcs/AmanatidisBFHV21}.\footnote{Prior work by Amanatidis et al.~\cite{journals/tcs/AmanatidisBFHV21} and Caragiannis et al.~\cite{conf/ec/CaragiannisGH19} showed that MNW allocations are not necessarily EFX for tri-valued instances, we extend this result by showing that no constant-factor approximation is achievable.}
Notably, our counterexample requires only three values, while their construction employs five distinct values.

\begin{example}
    Consider an instance with $n = 2$ agents and $m = 3$ goods $\{g_1, g_2, g_3\}$. 
    The agents' additive valuations are detailed in Table~\ref{table: tri_efx}, where $k$ is a sufficiently large number. 
    It can be verified that the allocation $\mathbf{X}$ given by $X_1 = \{g_1, g_2\}$ and $X_2 = \{g_3\}$ maximizes the Nash social welfare.
    However, consider agent 2, the EFX approximation ratio is bounded by:
    $$
        \frac{v_2(X_2)}{v_2(X_1 - g_2)} = \frac{1}{k} \xrightarrow[k\to\infty]{} 0,
    $$
    which implies that MNW allocations fail to provide any constant approximation guarantee for EFX in this instance.
\begin{table}[h]
    \centering
    \caption{Upper bound for EFX under tri-valued instances} 
    \label{table: tri_efx}
    \renewcommand{\arraystretch}{1.2} 
    \newcommand{\circval}[1]{\tikz[baseline=-0.6ex]{\node[draw,circle,inner sep=1.5pt] {$#1$};}} 
    \setlength{\tabcolsep}{12pt} 
    \begin{tabular}{lccc}
    \toprule
                & $g_1$ & $g_2$ & $g_3$\\
    \midrule
    Agent 1     & \circval{k} & \circval{1} & $1/k$ \\
    Agent 2     & $k$          &$1/k$           & \circval{1} \\
    \bottomrule
    \end{tabular}
\end{table}

\end{example}

Next, we address the upper bound for MMS.
Our construction adapts the approach of Caragiannis et al.~\cite{journals/teco/CaragiannisKMPS19} to the tri-valued setting.
Fix $n \geq 2$, let $\pi_n = 2/(1 + \sqrt{4n - 3})$ and $k = 1/\pi_n$.
Note that $k$ satisfies $k^2 - k - (n - 1) = 0$ by rearranging, i.e., $k = 1 + (n-1)/k$.

In the following, we construct a tri-valued instance in which we have $v_i(g) \in \{1/k, 1, k\}$ for every agent $i \in N$ and item $g \in M$.
Let the goods be $M = \{x\} \cup \{l_i, s_i: i = 2, 3, \dots, n\}$, resulting in a total of $|M| = 2n - 1$ items, where $l_i$ and $s_i$ denote the large item (with value $k$) and the small item (with value $1$) for every agent $i \geq 2$, respectively.
The other items are tiny (having value $1/k$) to agent $i$.
Agent $1$ values item $x$ at $1$, $l_i$ at $k$ and $s_i$ at $1/k$, for all $i\geq 2$ (see details in Table~\ref{table: mms_bound}).
Consider the allocation $\mathbf{X}$ where $X_1 = \{x\}$ and $X_i = \{l_i, s_i\}$ for every $i \geq 2$. 
It can be verified that this allocation maximizes the Nash social welfare.

For agent 1, we have $\MMS_1 = k$. 
This is because the goods can be partitioned into $n$ bundles each worth exactly $k$: $n-1$ bundles with a single large item $\{l_i\}$ for $i \geq 2$ (value $k$) and one bundle $\{x, s_2, \dots, s_n\}$ (value $1 + (n-1)/k = k$).
Consequently, the approximation ratio is:
$$
    \frac{v_1(X_1)}{\MMS_1} = \frac{1}{k} = \pi_n = \frac{2}{1 + \sqrt{4n - 3}}.
$$

This implies that MNW allocations cannot guarantee an MMS approximation ratio better than $2/(1 + \sqrt{4n - 3})$ for tri-valued instances.
\begin{table}[h]
    \centering
    \caption{Upper bound for MMS under tri-valued instances} 
    \label{table: mms_bound}
    \renewcommand{\arraystretch}{1.3} 
    \newcommand{\circval}[1]{\tikz[baseline=-0.6ex]{\node[draw,circle,inner sep=1.5pt] {\small $#1$};}}
    \setlength{\tabcolsep}{6pt} 
    \begin{tabular}{lcccccc}
    \toprule
                & $x$ & $l_2$ & $s_2$ & $\cdots$ & $l_n$ & $s_n$ \\
    \midrule
    Agent 1     & \circval{1} & $k$ & $1/k$ & $\cdots$ & $k$ & $1/k$ \\
    Agent 2     & $1/k$ & \circval{k} & \circval{1} & $\cdots$ & $1/k$ & $1/k$ \\
    $\vdots$    & $\vdots$ & $\vdots$ & $\vdots$ & $\ddots$ & $\vdots$ & $\vdots$ \\
    Agent $n$   & $1/k$ & $1/k$ & $1/k$ & $\cdots$ & \circval{k} & \circval{1} \\
    \bottomrule
    \end{tabular}
    %
\end{table}

Finally, note that PMMS and MMS are equivalent when $n=2$.
Therefore, the above result immediately implies an upper bound of $\pi_2 = 2/(1+\sqrt{5})\approx 0.618$ for the approximation ratio of MNW to PMMS.
\medskip

The construction above yields the following theorem for tri-valued instances.
\begin{theorem} \label{theorem: trivalue}
    For tri-valued instances, the approximation guarantees of MNW allocations with respect to EFX, MMS, and PMMS coincide with those established for general instances.
\end{theorem}

\section{Conclusion and Open Problems}

In this paper, we explore the boundaries of the ``unreasonable fairness'' of MNW allocations. 
Our investigation reveals a sharp contrast: while tri-valued instances exhibit the same theoretical limitations as the general setting, the bi-valued landscape offers significantly stronger guarantees.
By establishing tight characterizations for EFX, PMMS, and GMMS in the bi-valued domain, we systematically chart the landscape of these fairness guarantees, extending the insights of Caragiannis et al.~\cite{journals/teco/CaragiannisKMPS19}. 
The comparison with leximin and maximum $p$-mean welfare for $p\leq 1$ and $p\neq 0$ further highlights that these guarantees are specific to MNW and do not automatically extend to nearby egalitarian objectives.
A natural next step is to move from existence to computation, particularly focusing on finding allocations that satisfy these strengthened EFX, GMMS, and PMMS guarantees simultaneously, with or without PO.

\section*{Declaration on the Use of Generative AI}

All central research ideas, problem formulations, and mathematical results presented in this paper originated with the authors.
Every theorem, proof, and technical claim included in the final manuscript was independently verified by the authors, who take full responsibility for the correctness, originality, and presentation of the work.

\bibliographystyle{abbrv}
\bibliography{efx}

\newpage
\appendix

\section{Leximin Allocations for Bi-valued Instances} 
\label{section:leximin}

In this section, we turn our attention to leximin allocations, which is formally defined as one that lexicographically maximizes the vector of agents' utilities sorted in non-decreasing order.
Specifically, we examine their approximation guarantees regarding EFX, GMMS, and PMMS for bi-valued instances.
Similarly, we assume $m \geq n$ in the following analysis.
We show that even for the most restrictive factored instances, the approximation ratio by the leximin allocation is bad.
We demonstrate that the fairness guarantees of leximin allocations become arbitrarily bad when $k$ is large. 
This finding stands in sharp contrast to the strong guarantees provided by MNW allocations.

\begin{theorem}\label{theorem:leximin}
For factored bi-valued instances, every leximin allocation is $1/k$-EFX, $1/k$-GMMS, and $2/(k+1)$-PMMS, and the approximation guarantees are tight.
\end{theorem}

\subsection{Fairness Guarantee for EFX}

We first consider the EFX guarantee of leximin allocations, showing that leximin allocation is $1/k$-EFX for factored bi-valued instances.

\begin{lemma}\label{lemma:leximin_EFX}
For factored bi-valued instances, every leximin allocation is $1/k$-EFX, and the approximation guarantee is tight.
\end{lemma}

We first show that there exists a bi-valued instance for which the leximin allocation is at most $1/k$-EFX. 
Consider an instance with $n = 2$ agents and $m=t(k+1)$ items, where $t$ is a sufficiently large number. 
Agent 1 values every item at $1$, while Agent 2 values every item at $k$.

\begin{table}[h]
    \centering
    \caption{Hard instance for the approximation of EFX by leximin} 
    \label{table: efx_bound_upper_1k}
    \renewcommand{\arraystretch}{1.3}

    \newcommand{\circval}[1]{\tikz[baseline=-0.6ex]{\node[draw,circle,inner sep=1.5pt] {\small $#1$};}}

    \setlength{\tabcolsep}{6pt}

    \begin{tabular}{lcccccc}
    \toprule
              & $g_1$ & $g_2$ & $\cdots$ & $g_{tk}$ & $g_{tk+1}$ & $\cdots\quad g_{tk+t}$ \\
    \midrule
    Agent $1$ & \circval{1} & \circval{1} & $\cdots$ & \circval{1} & $1$ & $\cdots\quad 1$ \\
    Agent $2$ & $k$         & $k$         & $\cdots$ & $k$         & \circval{k} & $\cdots\quad \circval{k}$ \\
    \bottomrule
    \end{tabular}

    \raggedright \small
\end{table}

In the leximin allocation, agent $1$ receives $t k$ items and agent $2$ receives the remaining $t$ items, leading to $v_1(X_1) = v_2(X_2)=tk$.
However, the allocation is only $1/k$-EFX because
$$
\frac{v_2(X_2)}{v_2(X_1-g)}
= \frac{tk}{k(t k-1)}
= \frac{t}{t k-1}
= \frac{1}{k-\frac{1}{t}}
 \xrightarrow[t\to\infty]{}  \frac{1}{k}.
$$

Next, we show that any leximin allocation is at least $1/k$-EFX.
We begin with a structural observation about bundle sizes in leximin allocations.

\begin{observation}\label{observation:bundlesize}
Given any leximin allocation $\bX=(X_1,\dots,X_n)$, let $t=\min_{i\in[n]} v_i(X_i)$. Then for every agent $j\in N$, we have
$$
|X_j|\leq t+1.
$$
\end{observation}
\begin{proof}
Assume for contradiction that there exists an agent $j$ with $|X_j|\geq t+2$.
Let $i$ be an agent with minimum utility, i.e., $v_i(X_i)=t$.
Then by reallocating an item from $X_j$ to $X_i$, the utility of both bundles is at least $t+1$, contradicting the leximin optimality of the allocation.
\end{proof}

Therefore, for any two agents $i,j\in N$, and any item $g\in X_j$,
we have
\begin{equation*}
    v_i(X_i) \geq t \geq |X_j - g| \geq \frac{1}{k}\cdot v_i(X_j-g),
\end{equation*}
where the last inequality holds since every item has value at most $k$ to every agent.

Therefore, we obtain a tight characterization that leximin allocation is $1/k$-EFX.

\subsection{Fairness Guarantee for GMMS} \label{section: GMMS}

Next, we analyze the GMMS guarantee of leximin allocations and show that leximin provides a tight $1/k$-approximation to GMMS. 

\begin{lemma}\label{lemma:leximin_MMS}
For factored bi-valued instances, every leximin allocation is $1/k$-GMMS, and the approximation guarantee is tight.
\end{lemma}
Note that the leximin property is consistent under restriction to subsets of agents (otherwise, the global allocation could be strictly improved). 
Therefore, proving that any leximin allocation is $1/k$-MMS is sufficient to establish $1/k$-GMMS.

We start with the upper bound.
Consider an instance with $n$ agents and $m = n + kn(n-1)$ items. 
For $i\in[n-1]$, $v_i(g)=1$ for all $g\in M$, while $v_n(g)=k$ for all $g\in M$.
%
%
In the leximin allocation, we can guarantee that $v_i(X_i) = kn$ for all $i \in N$ (see Table~\ref{table: leximin_mms} for details). 
\begin{table}[h]
    \centering
    \caption{Hard instance for the approximation of MMS by leximin} 
    \label{table: leximin_mms}
    \renewcommand{\arraystretch}{1.3}

    \newcommand{\circval}[1]{\tikz[baseline=-0.6ex]{\node[draw,circle,inner sep=1.5pt] {\small $#1$};}}

    \setlength{\tabcolsep}{6pt}

    \begin{tabular}{lcccccccccc}
    \toprule
              & $g_1$ & $\cdots$ & $g_{kn}$ 
              & $g_{kn+1}$ & $\cdots$ & $g_{2kn}$ 
              & $\cdots$ 
              & $g_{kn(n-1)+1}$ & $\cdots$ & $g_m$ \\
    \midrule
    Agent $1$   & \circval{1} & $\cdots$ & \circval{1}
                & $1$ & $\cdots$ & $1$
                & $\cdots$
                & $1$ & $\cdots$ & $1$ \\
    Agent $2$   & $1$ & $\cdots$ & $1$
                & \circval{1} & $\cdots$ & \circval{1}
                & $\cdots$
                & $1$ & $\cdots$ & $1$ \\
    $\vdots$    & $\vdots$ & $\vdots$ & $\vdots$
                & $\vdots$ & $\vdots$ & $\vdots$
                & $\vdots$
                & $\vdots$ & $\vdots$ & $\vdots$ \\
    Agent $n-1$ & $1$ & $\cdots$ & $1$
                & $1$ & $\cdots$ & $1$
                & $\cdots$
                & $1$ & $\cdots$ & $1$ \\
    Agent $n$   & $k$ & $\cdots$ & $k$
                & $k$ & $\cdots$ & $k$
                & $\cdots$
                & \circval{k} & $\cdots$ & \circval{k} \\
    \bottomrule
    \end{tabular}

    \raggedright \small
\end{table}

Therefore, for agent $n$, her achieved MMS ratio is
\begin{equation*}
    \frac{v_n(X_n)}{\MMS_n}
     = 
    \frac{kn}{k(1+k(n-1))}
     = 
    \frac{n}{1+k(n-1)}
     = 
    \frac{1}{\frac{1}{n} + k\bigl(1-\frac{1}{n}\bigr)}
    \xrightarrow[n\to\infty]{}
    \frac{1}{k}. 
\end{equation*}

Thus, leximin allocation may not guarantee an approximation factor greater than $1/k$. 

Then, we show a matching lower bound.
Let $\bX=(X_1,\dots,X_n)$ be any leximin allocation.
By Observation~\ref{observation:bundlesize}, we have $v_i(X_i)\ge \left\lfloor \frac{m}{n}\right\rfloor$ for every agent $i\in N$.
On the other hand, we have $\MMS_i \leq k\cdot \left\lfloor \frac{m}{n}\right\rfloor$ because in any $n$-partition $P=(P_1,\dots,P_n)$ of $M$ there exists a bundle with no more than $\left\lfloor \frac{m}{n}\right\rfloor$ items.
Therefore, leximin allocation $\bX$ is $1/k$-MMS.

\subsection{Fairness Guarantee for PMMS}

Finally, we examine the PMMS guarantee of leximin allocations.
\begin{lemma} \label{lemma: leximin_pmms}
For factored bi-valued instances, every leximin allocation is $\frac{2}{k+1}$-PMMS, and the approximation guarantee is tight.
\end{lemma}

The upper bound corresponds to the GMMS hard instance with $n=2$, where the PMMS approximation ratio for agent $2$ is $\frac{2k}{k(2k+2)/2}=\frac{2}{k+1}$ (see Table~\ref{table: leximin_pmms} for details).

\begin{table}[h]
    \centering
    \caption{Hard instance for the approximation of PMMS by leximin} 
    \label{table: leximin_pmms}
    \renewcommand{\arraystretch}{1.3}

    \newcommand{\circval}[1]{\tikz[baseline=-0.6ex]{\node[draw,circle,inner sep=1.5pt] {\small $#1$};}}

    \setlength{\tabcolsep}{6pt}

    \begin{tabular}{lcccccc}
    \toprule
              & $g_1$ & $g_2$ & $\cdots$ & $g_{2k}$ & $g_{2k+1}$ & $g_{2k+2}$ \\
    \midrule
    Agent $1$ & \circval{1} & \circval{1} & $\cdots$ & \circval{1} & $1$ & $1$ \\
    Agent $2$ & $k$         & $k$         & $\cdots$ & $k$         & \circval{k} & \circval{k} \\
    \bottomrule
    \end{tabular}

    \raggedright \small
\end{table}

We now turn to the lower bound.
Consider an arbitrary leximin allocation $\mathbf{X}$. 
It suffices to show that for any two agents $i, j \in N$, $v_i(X_i) \geq 2\mu/(k+1)$, where $\mu = \MMS_i(T, 2)$ and $T = X_i \cup X_j$.
Since $\mu \leq v_i(T) / 2$, it is sufficient to prove that $v_i(X_j) \leq k\cdot v_i(X_i)$. 

Suppose that $v_i(X_j) > k \cdot v_i(X_i)$, which implies $|X_j| > v_i(X_i)$. 
By Observation~\ref{observation:bundlesize}, we also have $|X_j| \leq v_i(X_i) + 1$, forcing $|X_j| = v_i(X_i) + 1$.
We claim that $v_j(X_j) = |X_j|$. If not, $v_j(X_j) \geq |X_j| + 1 = v_i(X_i) + 2$.
Moving agent $j$'s least preferred item $g$ from $X_j$ to $X_i$ increases the utility of agent $i$.
For agent $j$, either $v_j(g) = 1$ (leaving $v_j(X_j - g) \geq v_i(X_i) + 1$) or $v_j(g) = k$ (leaving $v_j(X_j - g) = k v_i(X_i)$). 
Both cases yield a leximin improvement, leading to a contradiction.

Therefore, it suffices to consider the case where $v_j(X_j) = |X_j|$.
Since $v_i(X_j) > k \cdot v_i(X_i)$, there must exist $g \in X_j$ with $v_i(g) = k$ (otherwise $v_i(X_j) = v_i(X_i) + 1 > k \cdot v_i(X_i)$, contradicting $k \geq 2$).
Moving $g$ from $X_j$ to $X_i$ results in $v_i(X_i + g) = v_i(X_i) + k$ and $v_j(X_j -g) = v_i(X_i)$.
This improves the allocation in the leximin order, which leads to a contradiction.
Thus, $v_i(X_j) \leq k \cdot v_i(X_i)$, as claimed.

\section{Maximum \texorpdfstring{$p$}{p}-Mean Welfare for Bi-valued Instances}
\label{section:pmean}
In this section, we further show that the same negative phenomenon occurs for maximum $p$-mean welfare allocations whenever $p\leq 1$ and $p\neq 0$, under the standard convention that $p=0$ corresponds to MNW.

\begin{definition}[$p$-Mean Welfare]
For any $p \le 1$, the $p$-mean welfare of an allocation $\bX$ is defined as
$$
W_p(\bX)=
\begin{cases}
\left(\frac{1}{n}\sum_{i\in N} v_i(X_i)^p\right)^{\frac{1}{p}}, & p\neq 0,\\[4pt]
\left(\prod_{i\in N} v_i(X_i)\right)^{1/n}, & p=0.
\end{cases}
$$
For $p<0$, if some agent receives value $0$, we set $W_p(\bX)=0$.
An allocation that maximizes $W_p$ is called a \emph{maximum $p$-mean welfare allocation}.
In particular, maximum $0$-mean welfare allocations coincide with MNW allocations.
\end{definition}

\begin{theorem}\label{theorem:pmean_lower}
For every fixed $p\leq 1$ with $p\neq 0$, maximum $p$-mean welfare allocations do not guarantee any positive constant approximation for EFX, PMMS, and GMMS, even for factored bi-valued instances.
\end{theorem}

\begin{proof}
Since GMMS is stronger than PMMS, any upper bound on the PMMS approximation ratio also applies to the GMMS approximation ratio.
Thus, it suffices to establish upper bounds for EFX and PMMS.

We first consider the case when $p < 1$.
Fix an arbitrary $\epsilon\in(0,1)$.
We show that there is a factored bi-valued instance for which every maximum $p$-mean welfare allocation has EFX and PMMS approximation ratios below $\epsilon$.
Choose an integer $k\geq 2$ such that 
$$
q=k^{\frac{\lvert p\rvert}{1-p}}>\frac{4}{\epsilon} > 4,
$$
and an even integer $m$ such that $m\geq 4(1+q)$.
Consider an instance with two agents and $m$ identical items.
For every item $g\in M$, let
$v_1(g)=1$ and $v_2(g)=k$.
This is a factored bi-valued instance.
As $m$ is even and all items are identical for each agent, we have $
\MMS_1(M,2)= m / 2$ and 
$\MMS_2(M,2)=(km)/2$.

Every allocation is uniquely characterized by the number of items assigned to agent $1$. 
Specifically, if agent $1$ receives $x$ items, the values of agents $1$ and $2$ are $x$ and $k(m-x)$, respectively. 
To analyze the welfare, consider the function
$$
f(x) = x^p + \bigl(k(m-x)\bigr)^p.
$$
For $p<0$, maximizing $W_p$ is equivalent to minimizing $f(x)$ over $x\in\{1,\dots,m-1\}$ since the exponent $1/p$ is negative and allocations with zero value have $W_p=0$ by convention.
For $0<p<1$, maximizing $W_p$ is equivalent to maximizing $f(x)$ over $x\in\{1,\dots,m-1\}$.

The derivative
$$
f'(x) = p x^{p-1} - p k^p(m-x)^{p-1}
$$
is equal to zero when $x^{p-1} = k^p(m-x)^{p-1}$, which can be equivalently rewritten as
$$
\left(\frac{x}{m-x}\right)^{1-p} = k^{-p}.
$$
Since $x/(m-x)$ is strictly increasing on $(0,m)$, there is a unique solution.
We denote it by $x^*$, and it satisfies
$$
\frac{x^*}{m-x^*}=k^{-\frac{p}{1-p}}.
$$
The sign of $f'(x)$ shows that $f$ is decreasing before $x^*$ and increasing after $x^*$ when $p<0$, while $f$ is increasing before $x^*$ and decreasing after $x^*$ when $0<p<1$.
Furthermore, if $\bX$ is a maximum $p$-mean welfare allocation and $x'=|X_1|$, then we have $|x'-x^*|\leq 1$.
Below, we divide the analysis into two cases: $p < 0$ and $0 < p < 1$.

If $p<0$, then $x^*/(m-x^*)=q$, so we have 
$$
x^*=\frac{mq}{1+q}
\quad\text{and}\quad
m-x^*=\frac{m}{1+q}.
$$

In the following, we analyze the approximation ratio for agent $2$.
We begin by presenting two inequalities that will be used to bound the EFX ratio:
$$
m-x'\leq m-x^*+1\leq 2(m-x^*)
\quad\text{and}\quad
x'-1\geq x^*-2\geq \frac{x^*}{2}.
$$
The first inequality in each chain follows from $|x'-x^*|\leq 1$.
For the second inequalities, observe that $m - x^* = m / (1 + q) \geq 4$ and $x^* = mq / (1 + q) \geq 4$.

Therefore, the EFX approximation ratio for agent $2$ can be bounded by
$$
\frac{v_2(X_2)}{v_2(X_1-g)}
= \frac{m-x'}{x'-1}
\leq \frac{m-x^*+1}{x^*-2}
\leq \frac{2(m-x^*)}{x^*/2}
=\frac{4}{q}
<\epsilon.
$$
Similarly, the PMMS ratio satisfies
$$
\frac{v_2(X_2)}{\MMS_2(M,2)}
= \frac{k(m-x')}{km/2}
= \frac{2(m-x')}{m}
\leq \frac{2(m-x^*+1)}{m}
\leq \frac{4(m-x^*)}{m}
=\frac{4}{1+q}
<\epsilon.
$$
Hence, $\bX$ is neither $\epsilon$-EFX nor $\epsilon$-PMMS.

\smallskip

If $0<p<1$, then $x^*/(m-x^*)=1/q$, which implies
$$x^*=\frac{m}{1+q}
\quad\text{and}\quad
m-x^*=\frac{mq}{1+q}.
$$

In this case, we consider the approximation ratio for agent $1$.
To bound the EFX ratio, we utilize the following inequalities:
$$x'\leq x^*+1\leq 2x^*
\quad\text{and}
\quad
m-x'-1\geq m-x^*-2\geq \frac{m-x^*}{2}.
$$
The first inequality in each chain follows from $|x'-x^*|\leq 1$.
For the second inequalities, note that $x^*=m/(1+q)\geq 4$ and $m-x^*=mq/(1+q)\geq 4$.

Therefore, the EFX approximation ratio for agent $1$ is bounded by
$$
\frac{v_1(X_1)}{v_1(X_2-g)}
= \frac{x'}{m-x'-1}
\leq \frac{x^*+1}{m-x^*-2}
\leq \frac{2x^*}{(m-x^*)/2}
=\frac{4}{q}
<\epsilon.
$$
For the PMMS, we have 
$$
\frac{v_1(X_1)}{\MMS_1(M,2)}
= \frac{x'}{m/2}
\leq \frac{x^*+1}{m/2}
\leq \frac{2x^*}{m/2}
=\frac{4}{1+q}
<\epsilon.
$$
Again, $\bX$ is neither $\epsilon$-EFX nor $\epsilon$-PMMS.

As $\epsilon \in (0,1)$ can be arbitrarily small, no positive constant approximation guarantee can be achieved for any of these fairness notions.

\smallskip

Finally, we consider the case when $p=1$.
Consider an instance with two agents.
Choose any integer $k\geq 2$ and any even integer $m\geq 2$.
For every item $g\in M$, we let $v_1(g)=1$ and $v_2(g)=k$.
If agent $1$ receives $x$ items, the welfare objective is
$$
x+k(m-x)=km-(k-1)x,
$$
which is uniquely maximized at $x=0$.
Consequently, every allocation that maximizes the $1$-mean welfare must assign all items to agent $2$. 
Since agent $1$ receives an empty bundle, the resulting approximation ratios for EFX, PMMS, and GMMS are $0$.

\end{proof}

\section{Discussion of Personalized Instances with \texorpdfstring{$b_i = 0$}{bi = 0}} \label{section: discussion}
In this section, we present the discussion for agents with $b_i = 0$ under personalized instances. 
For any agent $i$ satisfying $b_i = 0$, we first normalize her valuation domain to $\{0, 1\}$. 
Since the MNW allocation inherently satisfies EF1, which implies GMMS under binary valuations, agent $i$ successfully achieves her GMMS guarantee in this allocation. 

Consequently, it suffices to focus on the EFX guarantee for these personalized instances. 
Prior to proceeding, we remark that items with zero value for all agents can be safely discarded from the initial allocation, as such items can be reallocated at the end to the source agent (which is not envied) in the envy graph without violating any fairness properties.

In our analysis, we have already shown that for any pair of agents $i$ and $j$ where $b_i > 0$ and $b_j > 0$, the allocation is $2/3$-EFX. 
Fix any pair of agents $(i, j)$. 
We discuss the EFX guarantee for two cases: (1) $b_i = 0$, and (2) $b_i > 0$ and $b_j = 0$.

When $b_i = 0$, we assume for the sake of contradiction that agent $i$ is not EFX towards $j$.
Then the EF1 property implies that $v_i(X_j) = v_i(X_i) + 1$, and there must exist some item $g \in X_j$ such that $v_i(g) = 0$.
Note that $v_j(g) \geq 1$ must hold. 
Otherwise, reallocating $g$ to some agent with non-zero value for $g$ yields a Nash welfare improvement.
Let $g' \in X_j$ be the item that minimizes $v_j(g')$ subject to $v_i(g') = 1$. 
Since $v_i(X_j) = v_i(X_i) + 1$, there are exactly $v_i(X_i) + 1$ items in $X_j$ satisfying the constraint.
This implies that $v_j(X_j) > (v_i(X_i) + 1) \cdot v_j(g')$. 
However, by Lemma~\ref{lemma:MNW_pb_dif}, we have $v_j(X_j) \le (v_i(X_i) + 1)\cdot v_j(g')$, which leads to a contradiction.

Next, we consider the case when $b_i > 0$ but $b_j = 0$.
Assume for the sake of contradiction that $i$ is not EFX towards $j$, there must exist both large and small items for $i$ in bundle $X_j$.
Since $b_j = 0$, agent $j$ values all these items at exactly $1$. Again, by Lemma~\ref{lemma:MNW_pb_dif}, we have $v_i(X_i) \ge k_i\cdot (v_j(X_j) - 1)$. Furthermore, for any item $g \in X_j$, it holds that $v_i(X_j - g) \le k_i \cdot (v_j(X_j) - 1) \le v_i(X_i)$. 
Therefore, $i$ is indeed EFX towards $j$, which is a contradiction. 
\end{document}